\documentclass[journal,10pt]{IEEEtran}  
\usepackage{cite}
\usepackage{indentfirst}
\usepackage{graphicx}
\usepackage{changepage}
\usepackage{stfloats}
\usepackage{amsfonts,amssymb}
\usepackage{bm}
\usepackage{algorithm}
\usepackage{algorithmic}
\usepackage{amsmath}
\usepackage{amsmath,amssymb,amsfonts}
\usepackage{times}
\usepackage{url}
\usepackage{cite}
\usepackage{textcomp}
\usepackage{xcolor}
\usepackage{color}
\usepackage{setspace}
\usepackage{enumitem}
\usepackage{float}
\usepackage{diagbox}
\usepackage[T1]{fontenc}
\usepackage[utf8]{inputenc}
\usepackage{authblk}
\usepackage{threeparttable}
\usepackage{booktabs}
\usepackage{footnote}
\usepackage{graphicx}
\usepackage{pifont}
\usepackage{subfigure}
\usepackage{mathrsfs}
\allowdisplaybreaks[4]

\newenvironment{proof}{{\indent  \indent \it Proof:}}{\hfill $\blacksquare$}

\begin{document}
\title{Sensing-Assisted Communication in Vehicular Networks with Intelligent Surface}

\author{
	Kaitao Meng, \textit{Member, IEEE}, Qingqing Wu, \textit{Senior Member, IEEE}, Wen Chen, \textit{Senior Member, IEEE}, and Deshi Li
	\thanks{K. Meng is the Department of Electronic and Electrical Engineering, University College London, London, UK, and also with with the State Key Laboratory of Internet of Things for Smart City, University of Macau, Macau, 999078, China (email: kaitao.meng@ucl.ac.uk). Q. Wu and W. Chen are with the Department of Electronic Engineering, Shanghai Jiao Tong University, Shanghai 201210, China (emails: \{qingqingwu, wenchen\}@sjtu.edu.cn). D. Li is with the Electronic Information School, Wuhan University, Wuhan, 430072, China. (email: dsli@whu.edu.cn). }
}
\maketitle
%%%%%%%%%%%%%%%%%%%%%%%%%%%%%%%%%%%%%%%%%%%%%%%%%%%%%%%%%%%%%%%%%%%%%%%%%%%%%%%%%

\begin{abstract}
 The recent development of integrated sensing and communications (ISAC) technology offers new opportunities to meet high-throughput and low-latency communication as well as high-resolution localization requirements in vehicular networks. However, considering the limited transmit power of the road site units (RSUs) and the relatively small radar cross section (RCS) of vehicles with random reflection coefficients, the power of echo signals may be too weak to be utilized for effective target detection and tracking. Moreover, high-frequency signals usually suffer from large fading loss when penetrating vehicles, which seriously degrades the communication service quality of users inside vehicles. To handle this issue, we propose a novel sensing-assisted communication scheme by employing an intelligent omni-surface (IOS) on the surface of vehicles to enhance both sensing and communication (S\&C) performance. To this end, we first propose a two-stage ISAC protocol, including the joint S\&C stage and the communication-only stage, to fulfil more efficient communication performance improvements benefited from sensing. The achievable communication rate maximization problem is formulated by jointly optimizing the transmit beamforming, the IOS phase shifts, and the duration of the joint S\&C stage. However, solving this ISAC optimization problem is highly non-trivial since inaccurate estimation and measurement information renders the achievable rate lack of closed-form expression. To handle this issue, we first derive a closed-form expression of the approximate achievable rate under uncertain location information, and then unveil a sufficient and necessary condition for the existence of the joint S\&C stage to offer useful insights for practical system design. Moreover, two typical scenarios including interference-limited and noise-limited cases are analyzed to provide a performance bound and a low-complexity algorithm for the considered systems. Finally, simulation results demonstrate the effectiveness of the proposed sensing-assisted communication scheme in achieving a higher achievable rate with lower transmit power requirements.
\end{abstract}   

\begin{IEEEkeywords}
	Intelligent surface, integrated sensing and communication, sensing-assisted communication, beamforming, phase shift design, vehicular networks.
\end{IEEEkeywords}
%%%%%%%%%%%%%%%%%%%%%%%%%%%%%%%%%%%%%%%%%%%%%%%%%%%%%%%%%%%%%%%%%%%%%%%%%%%%%%%%
\newtheorem{thm}{\bf Lemma}
\newtheorem{remark}{\bf Remark}
\newtheorem{Pro}{\bf Proposition}
\newtheorem{theorem}{\bf Theorem}
\newtheorem{Assum}{\bf Assumption}
\newtheorem{Cor}{\bf Corollary}

\section{Introduction}

Vehicle-to-everything (V2X) communications are expected to play an important role in next-generation wireless networks to support promising applications \cite{Challenges2021Gyawali, Interworking2016Abboud}, such as autonomous driving, traffic management, intelligent transportation, etc. In addition to providing high-quality information transmission, the environment sensing capability of vehicles is also of great importance due to the stringent requirements of positioning accuracy and latency in V2X networks, especially for high-density and high-mobility vehicular scenarios \cite{Huang2021MIMORadar, Bayesian2021Yuan}. Although the global positioning system (GPS) can provide mobile users with location and velocity information, it may not meet the accuracy requirement and could even become unavailable due to the potential signal block or interference, especially for urban environments with dense obstructions \cite{Liu2021Design}. Fortunately, recent advances in multiple-input and multiple-output (MIMO) and millimeter-wave (mmWave)/terahertz (THz) technologies offer new opportunities to simultaneously provide high-throughput and low-latency wireless communications, as well as ultra-accurate and high-resolution wireless sensing by sharing the same spectrum resources and wireless infrastructure \cite{Zhang2021Overview, Chiriyath2017Radar}. Following the above-mentioned advantages, the investigation on integrated sensing and communications (ISAC) technology is well underway \cite{Cui2021Integrating, Ericsson2020}. By simultaneously conveying information to the receiver and extracting information from the scattered echoes, ISAC offers exciting opportunities for higher spectral efficiency and lower hardware costs \cite{Yuan2021Integrated, meng2022throughput}. 

On the one hand, through the joint design of sensing and communication (S\&C) on the same infrastructure, ISAC can exploit the integration gain to achieve a flexible trade-off between S\&C over time-, frequency-, code-, and spatial-domain resources \cite{Ren2019Performance, Luong2021Radio}. To pursue such an integration gain, most existing works mainly focus on the investigation of waveform design, resource allocation, antenna deployment, etc. \cite{LiuX2020Joint, Liu2022OptimalBeamformer}. On the other hand, the synergy between S\&C offers the potential to achieve the coordination gain in V2X networks, e.g., sensing-assisted communication and communication-assisted sensing. In particular, some useful sensing-assisted communication mechanisms in vehicular networks are proposed to reduce the overhead caused by the high manoeuvrability of vehicles \cite{liu2020radar, du2021integrated, Yuan2021Bayesian}. For example, the authors in \cite{liu2020radar} proposed a novel extended Kalman filtering (EKF) framework to predict the kinematic parameters of vehicles and allocate the transmit power of the road site units (RSUs) according to the S\&C requirements, which helps reduce the overhead of the communication beam tracking. In \cite{du2021integrated}, the beamwidth was dynamically changed to exploit the beamforming gain of massive antennas of the RSU in a more flexible way, where a narrower beam is designed towards the communication receiver during the communication-only stage while a wider beam is utilized to cover the target during the sensing stage. A novel beamforming scheme based on probabilistic prediction was proposed in \cite{Yuan2021Bayesian}, where a distribution of the estimated parameters is constructed based on the echo signals and the state transition models of the vehicles. However, the echo signals may be too weak to achieve effective target detection and tracking, since the radar cross section (RCS) of the served vehicles in urban environments is generally small and the transmit power of RSUs would be limited. Moreover, high-frequency signals generally suffer from severe fading loss when penetrating into vehicles, which seriously degrades the communication performance of users inside vehicles.

%, Gao2022Beamforming, Chen2022ActiveIRS, Hua2021Intelligent
Recently, intelligent reflecting surfaces (IRSs) have been a promising technology to change signal propagation by exploiting the massive low-cost reflecting elements \cite{Zheng2022Survey, Chen2022ActiveIRS}. However, IRSs can only provide services when both the source and destination nodes lie on the same side of the IRS \cite{wu2019beamformingDiscrete, Hua2021Intelligent}. To address this, intelligent omni-surfaces (IOSs) are proposed to achieve more flexible S\&C coverage \cite{Xu2022Simultaneously, Huang2022Transforming}. Specifically, the signal received at the IOS can be reflected towards the incident side and/or refracted towards the other side of the IOS, thereby achieving a more flexible way to reconfigure wireless channels \cite{Simultaneously2022Mu}. From the perspective of target sensing, it is also demonstrated that the RCS of the IOS can be effectively changed by controlling its phase shifts \cite{zhang2018transmission}. Most existing works mainly focused on exploiting IOSs/IRSs to improve the communication performance of static or low-mobility users, where one or more surfaces are deployed in fixed locations, such as buildings or billboards \cite{ Simultaneously2022Mu, Xu2022Channel}. For instance, in \cite{Xu2022Channel}, an IRS-aided channel estimation method was proposed for V2X networks, where one IRS is deployed on the roadside and one RSU periodically transmits pilot symbols to analyze the correlations of the time-varying frequency-selective channels. However, the communication performance under such deployment strategies may be constrained for high-mobility vehicles due to the severe penetration loss of signals and the limited coverage of IOSs. To resolve this issue, in a most recent work \cite{Huang2022Transforming}, an intelligent refracting surface was deployed on the surface of a high-speed vehicle to aid its data transmission, where a training-based channel estimation technique is adopted to transform the channels from fast fading to slow fading for more reliable transmission. However, the signalling overhead for joint estimation of the massive MIMO channel matrix parameters is relatively high, which would induce transmission delays and limit its practical application in vehicular networks. Overall, how to efficiently reduce signalling overhead and provide stable and efficient data transmission in high-mobility communication scenarios is still an open issue.

\begin{figure}[t]
	\centering
	\includegraphics[width=8.7cm]{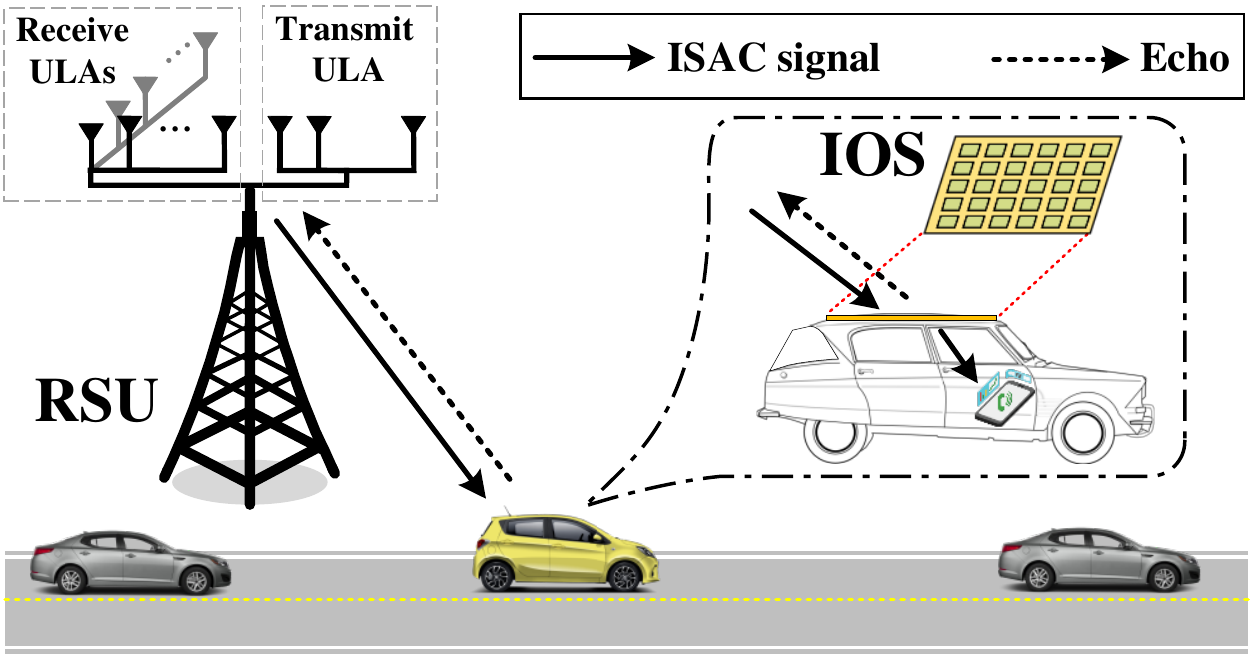}
	\caption{Scenarios of sensing-assisted vehicular
		communication.}
	\label{figure1a}
\end{figure}
In this paper, we propose a novel sensing-assisted communication scheme by employing an IOS on the surface of vehicles to enhance both the S\&C performance, as shown in Fig.~\ref{figure1a}. On the one hand, deploying an IOS on the surface of vehicles can focus the refracted signals toward users inside the vehicles and thereby improve the communication performance; on the other hand, the boosted echo signals reflected by the IOSs can be exploited at the RSU receivers to enhance the detection and tracking performance, and the controllability of echo signals can also reduce potential interference to other wireless systems in the same frequency. Furthermore, by utilizing the sensing results extracted from reflected echoes, the transmit beamforming and IOS phase shifts can be designed more effectively. To fulfil more effective communication improvements brought by sensing, we propose a two-stage ISAC protocol, including the joint S\&C stage and the communication-only stage. Specifically, during the joint S\&C stage, the IOS divides the signal power incident the IOS elements towards the RSU and the communication user respectively, so as to improve both S\&C services; during the communication-only stage, the signal power is concentrated to the communication user with the exploitation of the sensing results obtained in the former stage. This naturally leads to a fundamental trade-off between sensing accuracy and achievable rate. Specifically, if the IOS phase shift design aims at the maximization of the RSU's received power for a longer duration, the estimation/prediction accuracy is higher, which could benefit the subsequent communication; whereas this inevitably reduces the remaining time for communication services. On the other hand, higher signal power reflected towards the RSU can enhance the sensing accuracy and further improve the communication performance during the communication-only stage, while the achievable rate during the former stage inevitably decreases. Thus, it is important to properly allocate time and power resources to achieve an effective balance between S\&C and further provide a better sensing-assisted communication performance improvement.

To investigate the communication improvement brought by sensing, the achievable rate is maximized by jointly optimizing the transmit beamforming, the IOS phase shifts, and the duration of the joint S\&C stage. However, solving this ISAC optimization problem is highly non-trivial since there is no closed-form expression of the rate due to the inaccurate location information of vehicles. To handle this issue, we first derive a closed-form expression under a sufficiently large number of IOS elements, which is verified to achieve a good approximation by Monte Carlo simulations under practical setups. Furthermore, we present analysis for two typical scenarios including interference-limited and noise-limited cases to provide useful insights. The main contributions of this paper can be summarized as follows:
\begin{itemize}
	\item We propose a novel IOS-aided sensing-assisted communication scheme for vehicle communication systems, where an IOS is deployed on vehicles to enhance communication performance by jointly optimizing the transmit beamforming, the phase shift vectors, and the joint S\&C stage ratio. Based on the proposed two-stage protocol, the state estimation and measurement results during the joint S\&C stage are effectively utilized for communication improvement during the communication-only stage.
	\item We derive a closed-form expression of the achievable rate under uncertain angle information, where the relationship between beamforming gain and angle variance is provided. By properly allocating time and power resources to achieve an efficient balance between S\&C, communication improvements brought by sensing can be extremely enhanced. We further derive a sufficient and necessary condition for the existence of the joint S\&C stage, which can offer useful insights for practical system design.
	\item We propose an alternating optimization (AO) algorithm to jointly optimize the transmit power, the IOS phase shift, and the S\&C stage ratio for multi-vehicle cases. Then, two typical scenarios including interference-limited and noise-limited cases are analyzed to provide a performance bound for the considered systems.
	\item Finally, our simulation results verify the trade-off between S\&C performance for IOS-enabled ISAC systems and validate the superiority of the proposed schemes over benchmark schemes. The results reveal that deploying IOS on vehicles can effectively improve the S\&C coverage of the RSU and reduce the requirements of transmit power. 
\end{itemize}

\textit{Notations}: ${\rm{diag}}({\bm{x}})$ denotes a diagonal matrix whose main diagonal elements are the elements of ${\bm{x}}$. For a general matrix ${\bm{X}}$,  ${\bm{X}}^\dag$, ${\bm{X}}^T$, and ${\bm{X}}^H$ respectively denote its conjugate, transpose, and conjugate transpose. $\mathbb{E}_x[\cdot]$ denotes statistical expectation on the random variable $x$. $j$ in $e^{j \theta}$ denotes the imaginary unit. The distribution of a circularly symmetric complex Gaussian (CSCG) random variable with mean $x$ and variance $\sigma^2$ is denoted by $\mathcal{C} \mathcal{N}(x,\sigma^2)$.

\section{System Model and Problem Formulation}
As shown in Fig.~\ref{figure1a}, we consider an IOS-aided sensing-assisted system, where one RSU provides ISAC services for $K$ IOS-mounted mobile vehicles, indexed by $k \in {\cal{K}} = \{1,\cdots,K\}$. Without loss of generality, it is assumed that the RSU employs a general uniform linear array (ULA) with $M_t$ transmit antennas along the $x$-axis, and two perpendicular ULAs with $M_r$ receive antennas along the $x$- and $y$-axis, respectively, as shown in Fig.~\ref{figure1a}, and the vehicles drive along a straight road that is parallel to the $x$-axis.\footnote{By adding a fixed offset to the tracked angles based on the relative location information between the road and the ULA, the proposed scheme can be readily extended to other roads with general geometry.} The IOS is deployed on the top surface of each vehicle to replace the metal panel incurring high penetration loss. In this case, signals arriving at the IOS are partially refracted through the vehicle to improve the communication performance of a device (e.g., an onboard communication device or a mobile phone carried by a passenger), while the remaining parts are reflected towards the RSU to improve the sensing performance (e.g., angle estimation and target localization). In the considered system, the motion parameters, e.g., angles, distances, and velocities of the vehicles, can be estimated by analyzing the echo signals reflected from the IOSs. These parameters can be deemed as functions of time $t \in [0, T]$, with $T$ being the maximum served duration within the available coverage of the RSU. For notational convenience, the time period $T$ is divided into multiple small time slots with equal lengths of $\Delta T$, where the time slot is indexed by $n \in {\cal{N}} = \{1, \cdots, N\}$, and $N = \frac{T}{\Delta T}$. In general, the time slot length $\Delta T$ is set short enough, and thus it is practically assumed that the motion parameters keep constant within each time slot  \cite{Jayaprakasam2017Robust, liu2020radar}. Specifically, during the $n$th time slot, the vehicle $k$'s direction is defined by $\{\psi^x_{k,n}, \psi^z_{k,n}\}$, where $\psi^x_{k,n}$ and $\psi^z_{k,n}$ are respectively the azimuth and elevation angles of the geometric path connecting the IOS of vehicle $k$ (namely IOS $k$) and the RSU, as shown in Fig.~\ref{figure1b}. The Doppler frequency and the round-trip delay of echo signals reflected from IOS $k$ are denoted by $\mu_{k,n}$ and $\tau_{k,n}$, respectively. 

\begin{figure}[t]
	\centering
	\includegraphics[width=6.8cm]{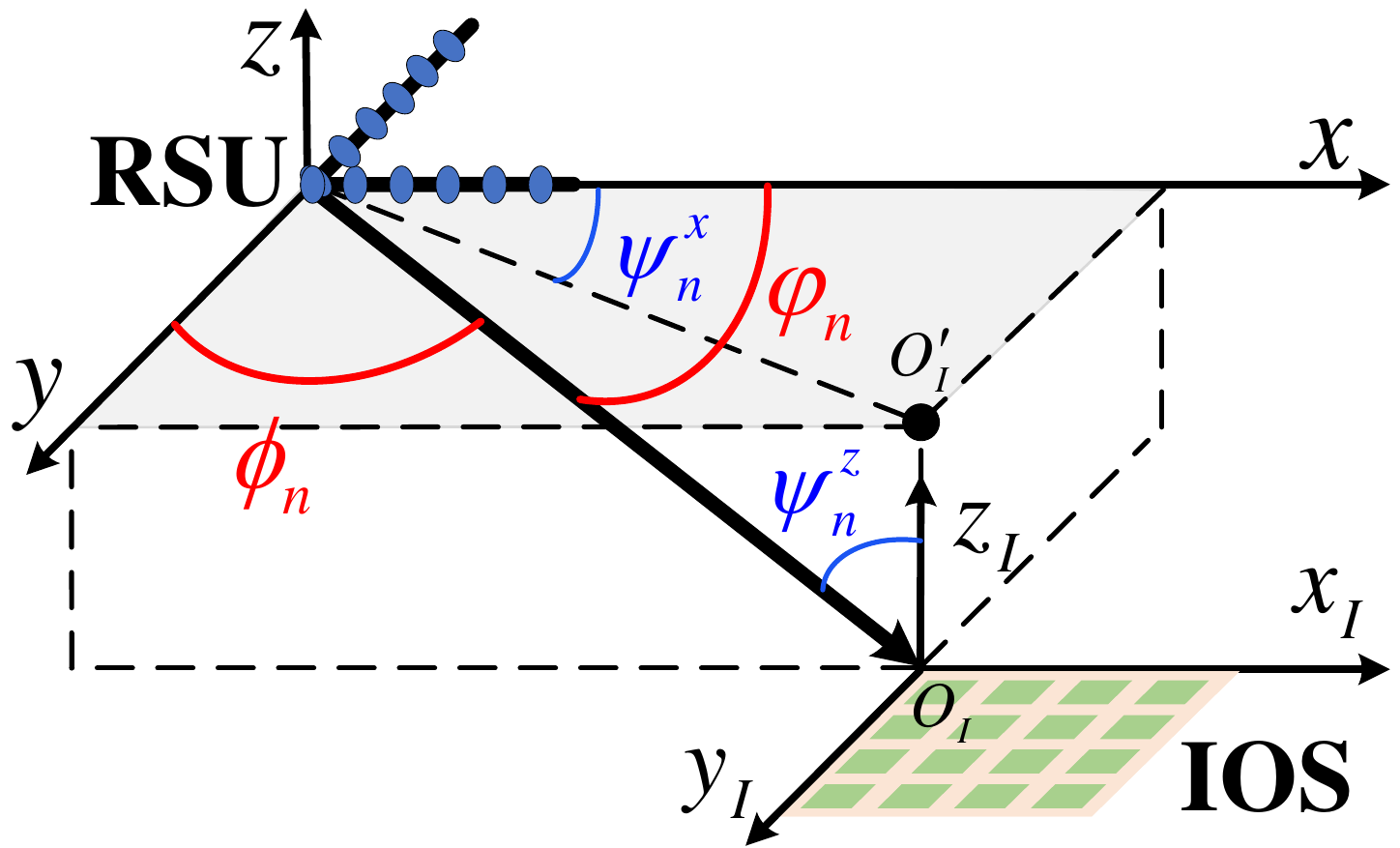}
	\caption{Description of the coordinate system. $O'_I$ is the projection point of $O_I$ on the $x-y$ plane at the coordinate of the RSU.}
	\label{figure1b}
\end{figure}

\subsection{Channel Model}
As shown in Fig.~\ref{figure1a}, the estimated angles of vehicle $k$ relative to the ULAs along the $x$- and $y$-axis are denoted by $\varphi_{k,n}$ and $\phi_{k,n}$, respectively, which are transmitted to the IOS controller for phase shift design.\footnote{It is worth noting that the estimated angles are not the exact azimuth and elevation angles between the IOS and the RSU (i.e., $\psi^x_{k,n}$ and $\psi^z_{k,n}$), which is illustrated in Fig.~\ref{figure1b}, and the relationship of these angle values is given below (\ref{SteeringVector}).} A uniform planar array (UPA) is equipped at the IOS. The number of IOS elements equipped on each vehicle is denoted by $L = L_x \times L_y$, where $L_x$ and $L_y$ denote the number of elements along the $x$- and $y$-axis, respectively. The half-wavelength antenna spacing is assumed for the UPA and ULAs. The channel between the RSU and IOS $k$ follows the free-space pathloss model and the corresponding channel power gain can be expressed as $\beta_{G,k,n} = \beta_0 d_{k,n}^{-2}$, where $\beta_0$ is the channel power at the reference distance 1 meter (m), and $d_{k,n}$ represents the distance from the RSU to IOS $k$ during the $n$th time slot. $\bm{H}^{\mathrm{DL}}_{k,n} \in \mathbb{C}^{M_t \times L}$ is the downlink channel matrix from the RSU to IOS $k$, given by
\begin{equation}\label{SteeringVector}
	\bm{H}^{\mathrm{DL}}_{k,n}=\sqrt{\beta_{G,k,n}} \bm{a}_{\mathrm{IOS}}\left(\phi_{k,n}, -\varphi_{k,n}\right) \bm{a}^T _{\mathrm{RSU}}\left(\varphi_{k,n}\right),
\end{equation}
where ${\bm{a}}_{\mathrm{RSU}}(\varphi_{k,n}) \!=\! \left[1, \cdots, e^{ {-j  \pi(M_t-1) \cos(\varphi_{k,n}) }}\right]^T$, ${\bm{a}}_{\mathrm{IOS}}(\phi_{k,n}, -\varphi_{k,n}) = \left[1, \cdots, e^{ {j \pi(L_{x}-1)  \cos(\varphi_{k,n}) }}\right]^T  \otimes\left[1, \cdots, e^{ { -j \pi(L_{y}-1) \cos({\phi}_{k,n}) }}\right]^T$, $\cos\left(\varphi_{k,n}\right) = \sin\left(\psi^z_{k,n}\right) \cos \left(\psi^x_{k,n}\right)$, $\cos\left({\phi}_{k,n}\right) = \sin\left(\psi^z_{k,n}\right) \sin \left(\psi^x_{k,n}\right)$, $\otimes$ denotes the Kronecker product. 
$\bm{H}^{\mathrm{UL},x}_{k,n}$ and $\bm{H}^{\mathrm{UL},y}_{k,n} \in \mathbb{C}^{L \times M_r}$ are the uplink channel matrices from IOS $k$ to the receive antennas along the $x$- and $y$-axis, respectively, given by
\begin{equation}
	\bm{H}^{\mathrm{UL},x}_{k,n}=\sqrt{\beta_{G,k,n}}\bm{b}_{\mathrm{RSU}}\left(\phi_{k,n}\right) \bm{a}_{\mathrm{IOS}}^{T}\left(\phi_{k,n}, -\varphi_{k,n}\right), 
\end{equation}
\begin{equation} \bm{H}^{\mathrm{UL},y}_{k,n}=\sqrt{\beta_{G,k,n}}\bm{b}_{\mathrm{RSU}}\left(\varphi_{k,n}\right) \bm{a}_{\mathrm{IOS}}^{T}\left(\phi_{k,n}, -\varphi_{k,n}\right),
\end{equation}
where $ {\bm{b}}_{\mathrm{RSU}}\left(x\right) = \left[1, \cdots, e^{ {-j \pi\left(M_r-1\right)  \cos(x) }}\right]^T$. Since the IOS and the communication device inside the vehicle remain relatively stationary and the distance between them is short, the IOS-device channel changes much more slowly as compared to the RSU-IOS channel. Thus, it can be approximately considered to be the line-of-sight (LoS) channel during each transmission frame, which is given by ${\bm{h}}_k
	= \sqrt{\beta_{h,k}} [1, \cdots, e^{ {-j \pi(L_{x}-1) {\Phi}^u_k }}]^T  \otimes[1, \cdots, e^{ { -j \pi(L_{y}-1) {\Omega}^u_k }}]^T$,
where ${\Phi}^u_k \triangleq \sin \left(\psi^{u,z}_{k,n}\right) \cos \left(\psi^{u,x}_{k,n}\right)$, ${\Omega}^u_k \triangleq \sin \left(\psi^{u,z}_{k,n}\right) \sin \left(\psi^{u,x}_{k,n}\right)$, and $\beta_{h,k}= \beta_0 d_{h,k}^{-2}$, where $d_{h,k}$ represents the distance from IOS $k$ to communication device $k$. The device's direction is defined by $\{\psi^{u,x}_{k,n}, \psi^{u,z}_{k,n}\}$, where $\psi^{u,x}_{k,n}$ and $\psi^{u,z}_{k,n}$ are respectively the azimuth and elevation angles of the geometric path connecting IOS $k$ and communication device $k$, and they can be obtained based on the device location or conventional channel estimation methods \cite{Zheng2020Intelligent}.

\begin{figure}[t]
	\centering
	\includegraphics[width=8.7cm]{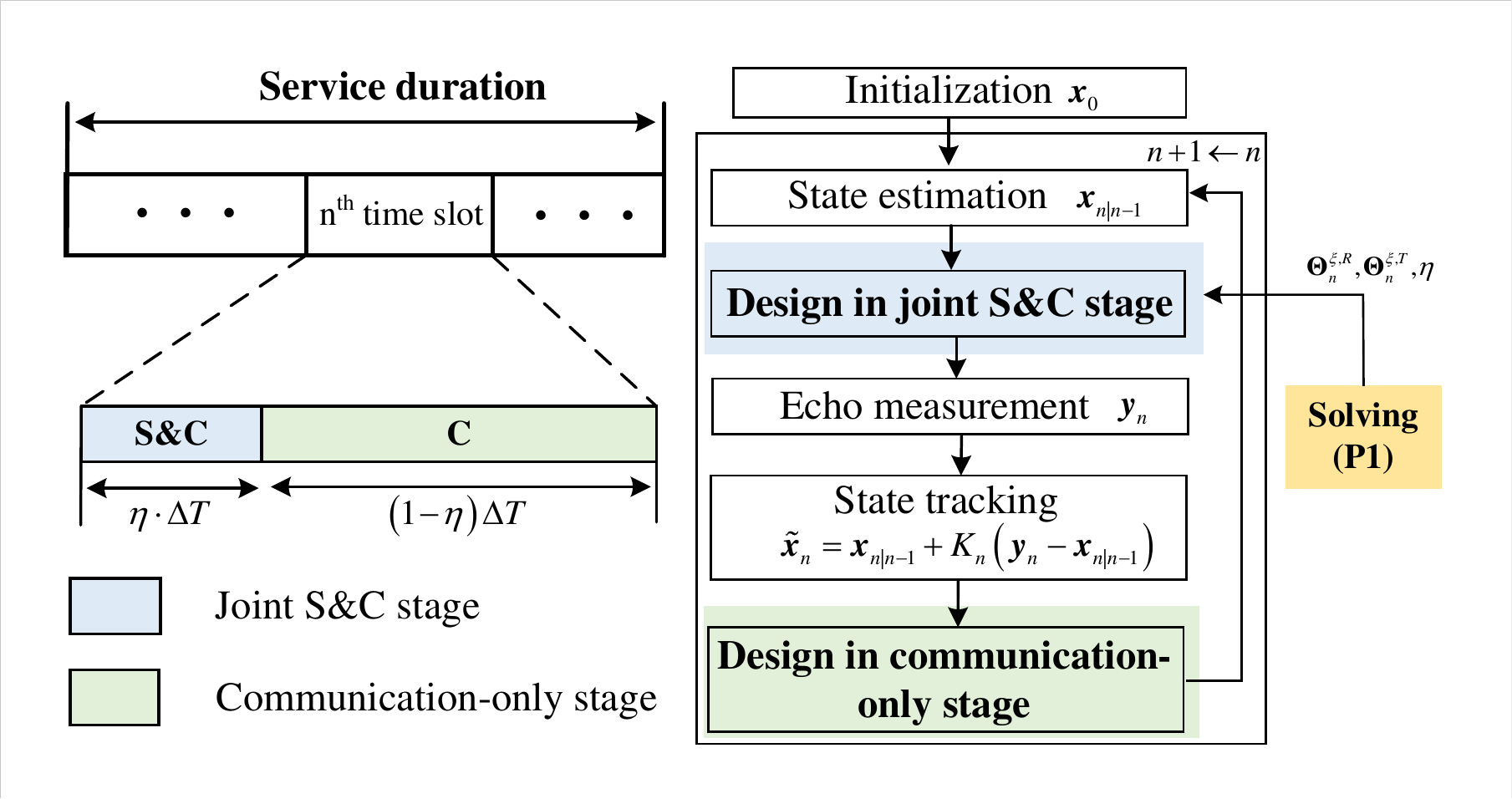}
	\caption{Illustration of the proposed two-stage sensing-assisted protocol.}
	\label{figure2}
\end{figure}
\subsection{Sensing-assisted Communication Framework}
\label{AOAEstimationError}
To effectively balance the S\&C performance and fully exploit the communication improvement brought by sensing, we design a two-stage sensing-assisted communication scheme to maximize the achievable rate of the high-mobility users. Specifically, each time slot is divided into two stages with a time splitting ratio of $\eta$ and $1-\eta$, where both S\&C services are provided in the former stage by splitting the signals incident on the IOS into two different directions, i.e., refracted signals and reflected signals, while the latter stage is solely designed for communication improvement by setting the IOS to be a totally refracting state, as shown in Fig.~{\ref{figure2}}. More specifically, during the joint S\&C stage, parts of the signals incident on the IOS are refracted into the vehicle for information converting while the others are reflected towards the RSU for the state measurements of the vehicles. After obtaining the measurements during the joint S\&C stage, the motion parameters in the current time slot can be estimated more efficiently to further improve beam tracking and communication performance during the communication-only stage. Specifically, for IOS $k$, the incident signal power on each element is generally split into two parts during the $n$th time slot: the power of the refracted and reflected signals with a splitting ratio of $\beta_{k,l,n}^{\xi,{\mathrm{T}}}$ and  $\beta_{k,l,n}^{\xi, {\mathrm{R}}}$, where $\xi \in \{\mathrm{\mathrm{S\&C}}, \mathrm{C}\}$ represents the state of the stages, i.e., the joint S\&C stage and the communication-only stage. Here, $\beta_{k,l,n}^{\xi,{\mathrm{T}}} + \beta_{k,l,n}^{\xi, {\mathrm{R}}} = 1$ and $\beta_{k,l,n}^{\xi,{\mathrm{T}}}, \beta_{k,l,n}^{\xi, {\mathrm{R}}} \in [0,1]$. Then, during the $\xi$ stage, the refraction- and reflection-coefficient matrices of IOS $k$ are given by ${\bm{\Theta}}^{\xi,{\mathrm{T}}}_{k,n} = {\rm{diag}}(\sqrt{\beta^{\xi,{\mathrm{T}}}_{k,1,n}}e^{j \theta^{\xi,{\mathrm{T}}}_{k,1,n}}, ... , \sqrt{\beta^{\xi,{\mathrm{T}}}_{k,L,n}}e^{j \theta^{\xi,{\mathrm{T}}}_{k,L,n}})$ and ${\bm{\Theta}}^{\xi, {\mathrm{R}}}_{k,n} = {\rm{diag}}(\sqrt{\beta^{\xi, {\mathrm{R}}}_{k,1,n}}e^{j \theta^{\xi, {\mathrm{R}}}_{k,1,n}}, ... , \sqrt{\beta^{\xi, {\mathrm{R}}}_{k,L,n}}e^{j \theta^{\xi, {\mathrm{R}}}_{k,L,n}})$, where $\theta^{\xi,{\mathrm{T}}}_{k,l,n}$ and $\theta^{\xi, {\mathrm{R}}}_{k,l,n} \in [0, 2\pi)$, respectively denote the refracting and reflecting phase shifts of the $l$th element of IOS $k$, $l \in {\cal{L}} = \{ 1,\cdots,L\}$.\footnote{The state-of-the-art technology can successfully realize separately controlling  the reflection phase-shift and refraction phase-shift of each element \cite{yang2018design, cai2017high}, and there are some related work investigating how to effectively realize the phase shift control in practice \cite{zhang2018transmission, bao2021programmable}. It has been demonstrated that compared to independent phase shift control, there is only $1\%$ performance loss for the coupled phase shift design mechanism in reference \cite{wang2022coupled}. The IOS can be powered by a dedicated/onboard lithium batter, and its power consumption is negligible compared to the power consumption of vehicles.}

As illustrated in Fig.~\ref{figure2}, during the joint S\&C/communication-only stage, the RSU beamforming and IOS phase shift vectors are jointly designed based on the state estimation/state tracking results, i.e., the channel state information required in our proposed scheme is composed of estimation information (also named prediction information) and measurement information. Specifically, the state evolution model of vehicle $k$ can be given by \cite{liu2020radar}
\begin{equation}
	\left\{ \!\begin{array}{l}
		\varphi_{k,n|n-1}=\varphi_{k,n-1}+d_{k,n-1}^{-1} v_{k,n-1} \Delta T \sin \left(\varphi_{k,n-1}\right)+\omega_{\varphi} \\
		\phi_{k,n|n-1} = \phi_{k,n-1} + d_{k,n-1}^{-1}{{{v_{k,n - 1}}\Delta T}}\tan \left( {{\phi _{k,n - 1}}} \right) +\omega_{\phi} \\
		d_{k,n|n-1}=d_{k,n-1}-v_{k,n-1} \Delta T \cos \left(\varphi_{k,n-1}\right)+\omega_{d} \\
		v_{k,n|n-1}=v_{k,n-1}+\omega_{v}
	\end{array}\right. \!\!\! .
\end{equation}
The covariance matrix of $\omega_{\varphi}$, $\omega_{\phi}$, $\omega_{d}$, and $\omega_{v}$ is ${\bm{Q}}_{\omega} = {\rm{diag}}([\sigma^2_{\omega_{\varphi}}, \sigma^2_{\omega_{\phi}}, \sigma^2_{\omega_{d}}, \sigma^2_{\omega_{v}}])$. Let ${\bm{x}}_{k,n} = [\varphi_{k,n}, \phi_{k,n}, d_{k,n}, v_{k,n}]$. The state evaluation model can be given by ${\bm{x}}_{k,n|n-1}=\bm{g}\left({\bm{x}}_{k,n-1}\right)+\boldsymbol{\omega}$, where $\boldsymbol{\omega} = [\omega_{\varphi}, \omega_{\phi}, \omega_{d}, \omega_{v}]$. Then, Kalman filtering is adopted for beam prediction and tracking. By denoting the state variables as ${\bm{x}}_{k,n|n-1} = [\varphi_{k,n|n-1}, \phi_{k,n|n-1}, d_{k,n|n-1}, v_{k,n|n-1}]$ and the measured variables as ${\bm{y}}_{k,n} = [\hat \varphi_{k,n}, \hat \phi_{k,n}, \hat d_{k,n}, \hat v_{k,n}]$, at the $n$th time slot, the measurement model of vehicle $k$ can be formulated as $\bm{y}_{k,n}={\bm{x}}_{k,n}+\bm{z}_{k,n}$, where $\bm{z}_{k,n} \in \mathcal{C} \mathcal{N}(0, {\bm{Q}}_{z_{k,n}} )$ and ${\bm{Q}}_{z_{k,n}} = {\rm{diag}}([\sigma^2_{z_{\varphi_{k,n}}}, \sigma^2_{z_{\phi_{k,n}}}, \sigma^2_{z_{d_{k,n}}}, \sigma^2_{z_{v_{k,n}}}])$. Overall, the number of channel state information variables required per time slot is $8K$. Accordingly, the prediction MSE matrix can be expressed as 
${\bm{M}}_{k,n \mid n-1} \!=\! {\bm{G}}_{k,n-1} {\bm{M}}_{k,n-1} {\bm{G}}_{k,n-1}^H + {\bm{Q}}_{\omega}$, where $\bm{G}_{k,n-1}=\left.\frac{\partial \bm{g}}{\partial \bm{x}}\right|_{\bm{x}=\hat{\bm{x}}_{k,n-1}}$.
Finally, the state tracking variable is $\tilde {\bm{x}}_{k,n} = [\tilde \varphi_{k,n}, \tilde \phi_{k,n}, \tilde d_{k,n}, \tilde v_{k,n}]$, given by
\begin{equation}\label{StateTracking}
	\tilde{\bm{x}}_{k,n}={\bm{x}}_{k,n \mid n-1}+\bm{K}_{k,n}\left(\bm{y}_{k,n}-{\bm{x}}_{k,n \mid n-1}\right),
\end{equation}
where $\bm{K}_{k,n}=\bm{M}_{k,n \mid n-1} \left(\bm{Q}_z + \bm{M}_{k,n \mid n-1} \right)^{-1}$.

\vspace{-1mm}
\subsection{Radar Measurement Model}
\label{RadarMeasurement}
During the joint S\&C stage, the transmitted ISAC signals at the RSU can be expressed as a weighted sum of ISAC beams, i.e., ${\bm{s}}_n(t) = \sum\nolimits_{k = 1}^K{\bm{f}}_{k,n}^{\mathrm{{S\&C}}}  {s_{k,n}}(t)$, where $s_{k,n}(t)$ is the ISAC signals for communication device $k$. It is assumed that the ISAC signals $\{s_{k,n}(t)\}$ of different users are independent random variables, with zero mean and unit variance, i.e., $\mathbb{E}[s_{k,n}[n]s_{j,n}[n]] = 0$, where $k \ne j$, and $k, j \in {\cal{K}}$. 
To maximize the achievable rate of vehicles, a low-complexity beamforming method is adopted by aligning the transmit beamforming vectors towards vehicles based on the estimation and measurement results. Hence, the transmit beamforming vectors during the joint S\&C stage and the communication-only stage are respectively given by ${\bm{f}}_{k,n}^{\mathrm{{S\&C}}} = \sqrt{\frac{p^{\mathrm{{S\&C}}}_{k,n}}{M_t}} {\bm{a}}^\dag_{\mathrm{RSU}}( \varphi_{k,n|n-1} )$ and ${\bm{f}}_{k,n}^{\mathrm{C}} = \sqrt{\frac{p^{\mathrm{C}}_{k,n}}{M_t}} {\bm{a}}^\dag_{\mathrm{RSU}}(\tilde \varphi_{k,n} )$, 
where $\varphi_{k,n|n-1}$ denotes the estimated angle based on the state evolution model and $\tilde \varphi_{k,n}$ is the state tracking angle (c.f. (\ref{StateTracking})). Similarly, the receive beamforming filter of the ULAs along the $x$- and $y$-axis are respectively given by ${\bm{v}}_{k,n}^{x} =  \sqrt{\frac{1}{M_r}} {\bm{b}}_{\mathrm{RSU}}( \varphi_{k,n|n-1} )$ and ${\bm{v}}_{k,n}^{y} =  \sqrt{\frac{1}{M_r}} {\bm{b}}_{\mathrm{RSU}}( \phi_{k,n|n-1} )$. Overall, the transmit beamforming and receive filter at the RSU are designed based on the estimation and measurement results, respectively. At the $n$th time slot, the ULA antennas along the $x$-axis receives the echoes contributed by the IOSs, expressed as
\vspace{-1mm}
\begin{align}
	{\bm{r}}^x_{n}(t) =& \sum\nolimits_{k = 1}^K  \sum\nolimits_{j = 1}^K  e^{j2\pi \mu_{j,n} t} {\bm{H}}^{\mathrm{UL},x}_{j,n} {\bm{\Theta}}_{j,n}^{\mathrm{S\&C,R}} {\bm{H}}^{\mathrm{DL}}_{j,n} {\bm{f}}_{k,n}^{\mathrm{{S\&C}}} \nonumber \\
	& \times s_{k,n}(t - \tau_{k,n}) + {\bm{z}}_r(t),
	\vspace{-1mm}
\end{align}
where ${\bm{z}}_r(t)$ denotes the noise at the receive ULAs. 
The resolved ranges and velocities in the delay-Doppler domain can be determined by utilizing the standard matched-filtering techniques \cite{liu2020radar}. It is assumed that the vehicles can be distinguished in the delay-Doppler domain. Then, for the receive ULA along the $x$-axis, the output of the matched filter for vehicle $k$ is expressed as 
%\begin{equation}
\vspace{-1mm}
\begin{align}\label{MatchFiltering}
	\widetilde{\bm{r}}^x_{k,n} \!=& {\beta_{G,k,n}} \sqrt{W} ({\bm{v}}^x_{k,n})^{H} \bm{b}_{\mathrm {RSU}}\left(\varphi_{k,n}\right)  \bm{a}^{T}_{\mathrm{IOS}}\left(\phi_{k,n},-\varphi_{k,n}\right) \nonumber \\
	& \times {\bm{\Theta}} ^{{\mathrm{S\&C,R}}}_{k,n}\bm{a}_{\mathrm{IOS}}\left(\phi_{k,n}, -\varphi_{k,n}\right) \bm{a}_{\mathrm {RSU}}^{T}\left(\varphi_{k,n}\right) {\bm{f}}_{k,n}^{\mathrm{{S\&C}}} \nonumber  \\
	& \times {\delta}\left(\tau-\tau_{k,n}, \mu-\mu_{k,n}\right)+\widetilde{\bm{z}}_{r},
	\vspace{-1mm}
\end{align}
%\end{equation}
where  ${\delta}\left(\tau-\tau_{k,n} , \mu-\mu_{k,n}\right)$ is the normalized matched-filtering output function obtained by time and frequency reversing and conjugating its own waveform for the complex transmit ISAC signal $s_{k,n}(t)$. In (\ref{MatchFiltering}), $\widetilde{\bm{z}}_{r}$ denotes the normalized measurement noise, and $W$ denotes the matched-filtering gain, which is equal to the number of symbols used for matched-filtering. The angle $\varphi_{k,n}$ can be readily measured by the maximum likelihood estimation (MLE) \cite{li1993maximum} or super-resolution algorithms like multiple signal classification (MUSIC) \cite{schmidt1986multiple}. Besides, angle $\phi_{k,n}$ can be estimated based on the received echo at the ULA along the $y$-axis in a similar way. As given in Section \ref{AOAEstimationError}, the AoA estimation results $\hat \varphi_{k,n} = \varphi_{k,n} + z_{\varphi_{k,n}}$ and $\hat \phi_{k,n} = \phi_{k,n} + z_{\phi_{k,n}}$, where $z_{\varphi_{k,n}}$ and $z_{\phi_{k,n}}$ represents the angle estimation errors, $ z_{\varphi_{k,n}} \in \mathcal{C} \mathcal{N}(0, \sigma_{z_{\varphi_{k,n}}}^{2} )$ and $ z_{\phi_{k,n}} \in \mathcal{C} \mathcal{N}(0, \sigma_{z_{\phi_{k,n}}}^{2} )$. Generally, the estimation error is inversely proportional to the received echo power at the RSU \cite{Liu2022Survey}, i.e.,
\vspace{-1.5mm}
\begin{equation}\label{ReceivedPowerEquation}
	\sigma_{z_{\varphi_{k,n}}}^{2} \propto ({{\gamma^{{\mathrm{S}},x}_{k,n}{{\sin }^2}\varphi_{k,n} }})^{-1}, \ \sigma_{z_{\phi_{k,n}}}^{2} \propto ({{\gamma^{{\mathrm{S}},y}_{k,n}{{\sin }^2}\phi_{k,n} }})^{-1},
	\vspace{-1.5mm}
\end{equation}
where ${\gamma^{{\mathrm{S}},x}_{k,n}}$ and ${\gamma^{{\mathrm{S}},y}_{k,n}}$ respectively denote the signal-to-noise ratio (SNR) at the receive antennas along the $x$- and $y$-axis after match-filtering. Due to the uncertain estimation information, at the $n$th time slot, the SNR of the received echos at the $x$-axis ULA is given in expectation form, i.e.,
\vspace{-1.5mm}
\begin{align}\label{SensingReceivedPower}
	\gamma^{{\mathrm{S}},x}_{k,n} =& \frac{\eta \Delta T}{\Delta t \sigma_s^2} {\beta^2_{G,k,n}} \mathbb{E}_{\varphi_{k,n|n-1}, \phi_{k,n|n-1}}\bigg[ |\underbrace {({\bm{v}}^x_{k,n})^{H}	{\bm{b}}_{\mathrm{RSU}}\left(\varphi_{k,n}\right)}_{\text{Receive beamforming gain}} \nonumber\\
	&\times \underbrace {\bm{a}^{T}_{\mathrm{IOS}}\left(\phi_{k,n},-\varphi_{k,n}\right) {\bm{\Theta}}^{{\mathrm{S\&C,R}}}_{k,n} \bm{a}_{\mathrm{IOS}}\left(\phi_{k,n},-\varphi_{k,n} \right)}_{\text{IOS beamforming gain}} \nonumber \\
 & \times \underbrace  {{\bm{a}}_{\mathrm{RSU}}^T\left(\varphi_{k,n}\right){\bm{f}}_{k,n}^{{\mathrm{S\&C}}}}_{\text{Transmit beamforming gain}}|^{2} \bigg],  
 \vspace{-1.5mm}
\end{align}
where $\sigma_s^2$ is the noise power at the receive antennas of the RSU, $\Delta t$ is the duration of one symbol, $\frac{\eta \Delta T}{\Delta t}$ represents the number of symbols used for matched-filtering during the joint S\&C stage. In (\ref{SensingReceivedPower}), the SNR at the RSU can be deemed as three MIMO beamforming gains: the transmit beamforming gain of the RSU, the passive beamforming gain of the IOS, and the receive beamforming gain of the RSU. For notational convenience, variables $\{\varphi_{k,n|n-1}, \tilde \varphi_{k,n}\}$ and $\{\phi_{k,n|n-1}, \tilde \phi_{k,n}\}$ in expectation operation $\mathbb{E}[ \cdot]$ are replaced as $\varphi$ and $\phi$, respectively. Similarly,  $\gamma^{{\mathrm{S}},y}_{k,n} = \frac{W}{\sigma_s^2} \mathbb{E}_{\varphi, \phi} [ | ({\bm{v}}^y_{k,n})^{H} {\bm{H}}^{\mathrm{UL},y}_{k,n} {\bm{\Theta}}^{{\mathrm{S\&C,R}}}_{k,n} {\bm{H}}^{\mathrm{DL}}_{k,n} {\bm{f}}_{k,n}^{\mathrm{S\&C}}|^{2} ]$. 

\subsection{Communication Model}
The communication device inside each vehicle mainly receives signals via the RSU-IOS-device link, since other non-line-of-sight (NLOS) links between the RSU and communication devices are practically assumed to be negligible due to severe penetration loss, especially for mmWave signals. During the joint S\&C stage, the signal received at communication device $k$ can be expressed as $
	y^{\mathrm{S\&C}}_{k,n}(t) =   {\bm{h}}^T_k{\bm{\Theta}}_{k,n}^{\mathrm{S\&C,T}}\bm{H}^{\mathrm{DL}}_{k,n} {\bm{f}}_{k,n}^{\mathrm{S\&C}} s_{k,n}(t) +  \sum\nolimits_{j \ne k}^K  {\bm{h}}^T_k{\bm{\Theta}}_{k,n}^{\mathrm{S\&C,T}}\bm{H}^{\mathrm{DL}}_{k,n} {\bm{f}}_{j,n}^{\mathrm{S\&C}} s_{j,n}(t) +  z_k(t)$,
where ${{z}}_k(t)$ denotes the noise at the receive antennas.\footnote{It is assumed that the Doppler effect induced by the vehicle mobility could generally be compensated at the communication receivers, which is generally presented in vehicular communication networks \cite{alieiev2018predictive}} The signal-to-interference-plus-noise ratio (SINR) of user $k$ during the joint S\&C stage is denoted by $\gamma^{\mathrm{S\&C}}_{k,n}$, and the approximate achievable rate during this stage can be given by $R^{\mathrm{S\&C}}_{k,n}$, i.e.,
\begin{align}\label{AchievablerateSC}
	\bar R^{\mathrm{S\&C}}_{k,n} \!=& \mathbb{E}_{\varphi, \phi}\left[ \log _{2}(1+ \gamma^{\mathrm{S\&C}}_{k,n})\right] \\
 \stackrel{(a)}{\leqslant}\! & \log _{2}\!\left(\!1 \!+\!\frac{  \mathbb{E}_{\varphi, \phi}\left[ \left|{\bm{h}}^T_k{\bm{\Theta}}^{\mathrm{S\&C,T}}_{k,n} {\bm{H}}^{\mathrm{DL}}_{k,n} {\bm{f}}_{k,n}^{\mathrm{C}}\right|^{2}\right]}{ \sum\nolimits_{j \ne k}^K \mathbb{E}_{\varphi, \phi}\left[ \left|{\bm{h}}^T_k{\bm{\Theta}}^{\mathrm{S\&C,T}}_{k,n} {\bm{H}}^{\mathrm{DL}}_{k,n} {\bm{f}}_{j,n}^{\mathrm{C}}\right|^{2}\right] \!+\! \sigma_{c}^{2}}\!\right)\!\! \nonumber \\
 \triangleq  & R^{\mathrm{S\&C}}_{k,n}, \nonumber
\end{align}
where ($a$) holds based on Jensen's inequality, i.e., $\mathbb{E}\left[ \log(x) \right] \le \log\left(\mathbb{E}[x] \right)$. During the communication-only stage, the signal received at communication device $k$ can be written as $
	y^{\mathrm{C}}_{k,n}(t) =   {\bm{h}}^T_k{\bm{\Theta}}_{k,n}^{\mathrm{C,T}}\bm{H}^{\mathrm{DL}}_{k,n} {\bm{f}}_{k,n}^{\mathrm{C}} s_{k,n}(t) +  \sum\nolimits_{j \ne k}^K  {\bm{h}}^T_k{\bm{\Theta}}_{k,n}^{\mathrm{C,T}}\bm{H}^{\mathrm{DL}}_{k,n} {\bm{f}}_{j,n}^{\mathrm{S\&C}} s_{j,n}(t) + z_k(t)$,
where ${\bm{f}}_{k,n}^{\mathrm{C}}$ denotes the beamforming vector of the RSU during the communication-only stage at the $n$th time slot. Similarly, the approximate achievable rate is given by $R^{\mathrm{C}}_{k,n}$, i.e.,
\begin{align}\label{AchievablerateC}
\bar R^{\mathrm{C}}_{k,n} =&\mathbb{E}_{\varphi, \phi}\left[ \log _{2}(1+\gamma^{\mathrm{C}}_{k,n})\right] \\
\le& \log _{2}\left(1+\frac{ \mathbb{E}_{\varphi, \phi}\left[\left|{\bm{h}}^T_k{\bm{\Theta}}_{k,n}^{\mathrm{C,T}}{\bm{H}}^{\mathrm{DL}}_{k,n} {\bm{f}}_{k,n}^{\mathrm{C}}\right|^{2}\right]}{ \sum\nolimits_{j \ne k}^K \mathbb{E}_{\varphi, \phi}\left[ \left|{\bm{h}}^T_k{\bm{\Theta}}^{\mathrm{C,T}}_{k,n} {\bm{H}}^{\mathrm{DL}}_{k,n} {\bm{f}}_{j,n}^{\mathrm{C}}\right|^{2}\right] + \sigma_{c}^{2}}\!\right)\! \nonumber \\
\triangleq & R^{\mathrm{C}}_{k,n} , \nonumber
\end{align}
where $\gamma^{\mathrm{C}}_k$ denotes the SINR of communication device $k$ during the communication-only stage. Then, at the $n$th time slot, the approximate achievable rate of communication device $k$ is given by 
\begin{equation}
	R_{k,n} = \eta R_{k,n}^{{\mathrm{S\&C}}} + \left( 1 - \eta  \right) R_{k,n}^{\mathrm{C}}.
\end{equation}

\subsection{Problem Formulation}
At each time slot, we aim to maximize the minimum approximate rate among the communication devices in the considered system by jointly optimizing the beamforming vectors of the RSU, the phase shift matrices of IOSs, and the duration of the joint S\&C stage. The optimization problem is formulated as 
\begin{alignat}{2}
	\label{P1}
	(\rm{P1}): & \begin{array}{*{20}{c}}
		\mathop {\max }\limits_{{\bm{p}}^{\mathrm{S\&C}}_n, {\bm{p}}^{\mathrm{C}}_n, \{{\bm{\Theta}}^{\xi,{\mathrm{R}}}_{k,n}\}, \{{\bm{\Theta}}^{\xi,{\mathrm{T}}}_{k,n}\}, \eta} \quad  \mathop {\min }\limits_{k} \  R_{k,n}
	\end{array} & \\ 
	\mbox{s.t.}\quad
	& \theta^{\xi,{\mathrm{T}}}_{k,l,n}, \theta^{\xi,{\mathrm{R}}}_{k,l,n} \in [0, 2 \pi), \forall l, k, \xi, & \tag{\ref{P1}a}\\
	& \beta^{\xi,{\mathrm{T}}}_{k,l,n}, \beta^{\xi,{\mathrm{R}}}_{k,l,n} \in [0, 1], \beta^{\xi, {\mathrm{T}}}_{k,l,n} + \beta^{\xi, {\mathrm{R}}}_{k,l,n}  = 1,  \forall l, k, \xi, & \tag{\ref{P1}b}\\
	& \sum\nolimits_{k = 1}^K p _{k,n}^{\mathrm{S\&C}} \le P^{\max}, \tag{\ref{P1}c} \\
	& \sum\nolimits_{k = 1}^K p _{k,n}^{\mathrm{C}} \le P^{\max}, \tag{\ref{P1}d} \\
	&  p _{k,n}^{\mathrm{S\&C}}\ge 0, p _{k,n}^{\mathrm{C}} \ge 0, \forall k, \tag{\ref{P1}e} \\
	& \eta \in [0, 1], & \tag{\ref{P1}f}
\end{alignat}
where ${\bm{p}}^{\mathrm{S\&C}}_n = [{{p}}_{1,n}^{\mathrm{S\&C}}, \cdots, {{p}}_{K,n}^{\mathrm{S\&C}}]$ and $ {\bm{p}}^{\mathrm{C}} = [{{p}}_1^{\mathrm{C}}, \cdots, {{p}}_K^{\mathrm{C}}]$. Solving (P1) optimally is non-trivial due to the lack of the closed-form objective function and closely coupled variables. To tackle this issue, we first derive an approximate achievable rate in closed form for single-vehicle cases in Section \ref{SingleVehicle}. Then, an efficient AO algorithm is proposed to solve the problem in multi-vehicle cases in Section \ref{MultiVehicle}.

\section{Single-vehicle Sensing-assisted Communication}
\label{SingleVehicle}
In this section, we consider the single-vehicle setup, i.e., $K=1$, to draw useful insights into the optimal IOS phase shift and S\&C duration ratio design. In this case, (P1) is simplified to (by dropping the communication device index) maximizing $R_n$ at the $n$th time slot, i.e., 
\begin{alignat}{2}
	\label{P2}
	(\rm{P2}): & \begin{array}{*{20}{c}}
		\mathop {\max }\limits_{{\bm{\Theta}}^{\xi, {\mathrm{R}}}_n, {\bm{\Theta}}^{\xi, {\mathrm{T}}}_n, \eta} \quad  R_n
	\end{array} & \\ 
	\mbox{s.t.}\quad
	& \theta^{\xi, {\mathrm{T}}}_{l,n}, \theta^{\xi, {\mathrm{R}}}_{l,n} \in [0, 2 \pi), \forall l, \xi, & \tag{\ref{P2}a}\\
	& \beta^{\xi, {\mathrm{T}}}_{l,n}, \beta^{\xi, {\mathrm{R}}}_{l,n} \in [0, 1], \beta^{\xi, {\mathrm{T}}}_{l,n} + \beta^{\xi, {\mathrm{R}}}_{l,n}  = 1, \forall  l, \xi,  & \tag{\ref{P2}b}\\
	& \eta \in [0, 1]. & \tag{\ref{P2}c}
\end{alignat}
However, it is challenging to solve (P2) optimally since there is no closed-form expression of $R_n$. To tackle this issue, we will first derive an approximate expression of $R_n$ in closed form, based on which, a sufficient and necessary condition for the existence of the joint S\&C stage will be presented to solve (P2) more efficiently.

\subsection{Closed-form Expression of Achievable Rate and Proposed Solution}

In this subsection, the optimal phase shift is given first, based on which, a closed-form expression of $R_n$ is derived.
\begin{thm}\label{OptimalPhaseShift}
	At the optimal solution of (P2), during the joint S\&C stage, the reflecting and refracting phase shifts, $\theta^{\mathrm{R}}_{(l_x-1)L_y+l_y}$ and $\theta^{\mathrm{T}}_{(l_x-1)L_y+l_y}$, of the $(l_x; l_y)$th (the $l_x$th row and the $l_y$th column) element of the IOS is obtained by
	\begin{equation}
		\begin{aligned}
		&\theta^{\mathrm{R}}_{l_x, l_y}=\pi(l_x-1) q^{\mathrm{R}}_{{x}}+\pi(l_y-1) q^{\mathrm{R}}_{{y}}+\theta_{0}, \\
		&\theta^{\mathrm{T}}_{l_x, l_y}=\pi(l_y-1) q^{\mathrm{T}}_{{x}}+\pi(l_y-1) q^{\mathrm{T}}_{{y}}+\theta_{0}, 
	\end{aligned}
	\end{equation}
	where {$\theta_{0}$} is the reference phase at the origin of the coordinates, $q^{\mathrm{R}}_{{x}} =  - 2 \cos( \varphi_{n|n-1})$, $q^{\mathrm{R}}_{{y}} =  2 \cos( \phi_{n|n-1})$, $q^{\mathrm{T}}_{{x}} = - \cos( \varphi_{n|n-1}) + \sin ( \psi^{u,z}_{n}) \cos( \psi^{u,x}_{n})$, and $q^{\mathrm{T}}_{{y}} = \cos( \phi_{n|n-1}) + \sin ( \psi^{u,z}_{n}) \sin( \psi^{u,x}_{n})$.
\end{thm}
\begin{proof}
	Please refer to Appendix A.
\end{proof}

According to Lemma \ref{OptimalPhaseShift}, the phase shifts of all IOS elements are designed to align the reflecting beam toward the RSU. Similar to Lemma \ref{OptimalPhaseShift}, the optimal phase gradients $q_{{x}}$ and $q_{{y}}$ during the communication-only stage are given by $q_{{x}} =  -\cos(\tilde \varphi_{n}) + \sin ( \psi^{u,z}_{n}) \cos( \psi^{u,x}_{n})$ and $q_{{y}} =  \cos(\tilde \phi_{n}) +  \sin ( \psi^{u,z}_{n}) \sin( \psi^{u,x}_{n})$.

\begin{thm}\label{EqualPowerSplitting}
	At the optimal solution of (P2), $\beta^{\xi, {\mathrm{R}}*}_{l,n} = \beta^{\xi, {\mathrm{R}}*}_{l',n}$, $\beta^{\xi,{\mathrm{T}}*}_{l,n} = \beta^{\xi,{\mathrm{T}}*}_{l',n}$, where $\xi \in \{\mathrm{S\&C, C}\}$, $l \ne l'$,  and $l, l' \in {\cal{L}}$. 
\end{thm}
\begin{proof}
	Please refer to Appendix B.
\end{proof}

According to Lemma \ref{EqualPowerSplitting}, all IOS elements share the same refracted and reflected power splitting ratio, denoted by $\beta^{\xi, {\mathrm{T}}}$ and $\beta^{\xi,\mathrm{R}}$. In particular, during the communication-only stage, $\beta^{\mathrm{C, R}} = 0$ and $\beta^{\mathrm{C, T}} = 1$. For notational simplification, $\beta^{\mathrm{S\&C, R}}$ and $\beta^{\mathrm{S\&C, T}}$ are respectively replaced by $\beta^{\mathrm{R}}$ and $\beta^{\mathrm{T}}$ in the following analysis. The transmit beamforming gain in (\ref{SensingReceivedPower}) can be expressed as
	\begin{align}
		&\left|{{\bm{a}}^H_{\mathrm{RSU}}\left( \varphi_n\right){\bm{f}}_{n}^{\mathrm{S\&C}}}\right|^2 =\frac{{P^{\max}}}{{{M_t}}}\left| {\sum\limits_{m = 1}^{{M_t}} {{e^{j\pi m\left( {\cos \left(\varphi_{n|n-1} \right) - \cos ( \varphi_n  )} \right)}}} } \right|^2 \nonumber\\
		& \buildrel \Delta \over = {P^{\max}} F_{M_t}\left( {\cos ( \varphi_{n|n-1} ) - \cos ( \varphi_n  )} \right) ,
	\end{align}
where the Fej$\acute{e}$r kernel $F_{M}(x) = \frac{1}{M}\left(\frac{\sin \frac{M \pi x}{2 }}{\sin \frac{\pi x}{2 }}\right)^{2}$ \cite{rust2013convergence}. Similarly, the receive beamforming gain in (\ref{SensingReceivedPower}) can be expressed as $\left|{\bm{v}}_{n}^{H} {{\bm{b}}_{\mathrm{RSU}}\left( \varphi_n\right)}\right|^2  = F_{M_r}( {\cos ( \varphi_{n|n-1} ) - \cos ( \varphi_n  )} ) $.

Based on Lemmas \ref{OptimalPhaseShift} and \ref{EqualPowerSplitting}, the passive beamforming gain achieved by the IOS is given by
%\begin{equation}
\begin{equation}
\begin{aligned}
	\left|\bm{a}^{T}_{\mathrm{IOS}}\left(\phi_n, -\varphi_n\right) {\bm{\Theta}} ^{{\mathrm{S\&C,T}}}_n\bm{a}_{\mathrm{IOS}}\left(\phi_n, -\varphi_n\right)\right|^2   \\
 \buildrel \Delta \over =   \beta^{\mathrm{R}} L_x  L_y F_{L_x}\left( 2\Delta \cos \varphi_n\right) F_{L_y}\left( 2\Delta \cos \phi_n\right) ,
\end{aligned}
\end{equation}
where $\Delta \cos \varphi_n =  \cos \left( \varphi_{n|n-1} \right) - \cos \left( \varphi_n  \right)$ and $\Delta \cos \phi_n =  \cos \left( \phi_{n|n-1} \right) - \cos \left( \phi_n  \right)$. For the receive ULA along the $x$-axis, the SNR of the received echos can be given by
\begin{align}\label{SensingPowerAverage}
	&\gamma^{{\mathrm{S}},x}_n =  \frac{\eta \Delta T \mathbb{E}_{\varphi, \phi}\left[ \left| ({\bm{v}}^x_{k,n})^{H} {\bm{H}}^{\mathrm{UL},x}_{n} {\bm{\Theta}}^{\mathrm{S\&C,T}}_n {\bm{H}}^{\mathrm{DL}}_{n} {\bm{f}}_{n}^{\mathrm{S\&C}}\right|^{2} \right]}{\Delta t \sigma_s^2} \nonumber \\
	= & \beta^{\mathrm{R}}\eta \frac{ L  \Delta T {P^{\max}} {\beta^2_{G,n}}}{\Delta t  {\sigma_s^2} } \mathbb{E}_{\varphi}[ F_{L_x}\left( 2\Delta \cos \varphi_n\right)F_{M_t}\left(\Delta \cos \varphi_n\right)  \nonumber \\
	&\times F_{M_r}\left(\Delta \cos \varphi_n\right)] \mathbb{E}_{\phi}\left[F_{L_y}\left( 2\Delta \cos \phi_n\right)  \right] . 
\end{align}
With the perfect angle information, i.e., $\Delta \cos \varphi_n =0$ and $\Delta \cos \phi_n = 0$, a power gain of order $M_t M_r L^2$ can be achieved. However, the sensing power in (\ref{SensingPowerAverage}) involves the expectation operation, which makes it challenging to optimize the time and power resources in the considered system. This thus motivates the following proposition. 

\begin{Pro}\label{InftyBand}
	When the number of IOS elements is sufficiently large, i.e., $L_x \to \infty$ and $L_y \to \infty$, the SNR expectation of received echoes at the ULAs along the $x$- and $y$-axis can be approximated as 
	\begin{equation}\label{ClosedFormSensingPower}
	\gamma^{{\mathrm{S}},x}_n \!\approx \!\gamma^{{\mathrm{S}},y}_n \!\approx\! \beta^{\mathrm{R}}\eta \frac{  \Delta T {P^{\max}}  \beta_{G,n}^2 L  M_t M_r h(\varphi_n, \sigma^2_{\omega_\varphi}) h(\phi_n, \sigma^2_{\omega_\phi}) }{{\Delta t \sigma_s^2} }  \! .
\end{equation}
where $h(x, y) \buildrel \Delta \over=\! \sum\limits_{i = -\infty}^{\infty} \frac{1}{{\sqrt {2\pi {y}\sin^2({x})}  }}\!\left(\! {{e^{ - \frac{{{{2( {i\pi} )}^2}}}{{{y}}}}} + {e^{ - \frac{{{{2\left( {(i + 1)\pi  } - x \right)}^2}}}{{{y}}}}}} \!\right)\!$.
\end{Pro}
\begin{proof}
	Please refer to Appendix C.
\end{proof}

Proposition \ref{InftyBand} presents a closed-form rate for the practical case with a large number of IOS elements, which achieves a good approximation and is later verified in Section \ref{simulationS} by the Monte Carlo simulation results. On the other hand, according to the analysis in (\ref{ReceivedPowerEquation}), we have
\begin{equation}
\sigma_{z_{\varphi_{k,n}}}^{2} = \frac{{\sigma _R^2}}{{\gamma^S_n{{\sin }^2}\varphi_n }} =  \frac{A_{\varphi_n}}{\eta \beta^{\mathrm{R}}},
\end{equation}
where $A_{\varphi_n} = \frac{  \Delta t \sigma_s^2 \sigma _R^2 }{{\Delta T {P^{\max}}  \beta_{G,n}^2 L  M_t M_r h(\varphi_n, \sigma^2_{\omega_\varphi}) h(\phi_n, \sigma^2_{\omega_\phi}) {{\sin^2 }}\varphi_n } }$, $\sigma _R^2$ is the variance parameter of the angle estimation method, and $A_{\varphi_n}$ represents the angle variance when $\eta = 1$ and $\beta^{\mathrm{R}} = 1$. Since $d \gg v \Delta T$, the angle variance of the state tracking error in (\ref{StateTracking}) can be given by
\begin{equation}\label{StateTrackingError}
	\sigma _{\tilde \varphi_n} ^2 = {\sigma _{{\omega_{\varphi}}} ^2} - \frac{\sigma _{{\omega_{\varphi}}} ^2}{\sigma _{{\omega_{\varphi}}} ^2 + \frac{A_{\varphi_n}}{\eta \beta^{\mathrm{R}}}} {\omega^2 _{\varphi}} = \frac{\sigma _{{\omega_{\varphi}}} ^2 A_{\varphi_n}}{\sigma _{{\omega_{\varphi}}} ^2 {\eta \beta^{\mathrm{R}}} + {A_{\varphi_n}}}.
\end{equation}
In (\ref{StateTrackingError}), it is observed that the variance of the state tracking error decreases monotonically with the variance of angle estimation $\sigma_{z_{\varphi_{k,n}}}^{2}$. Similarly, we have $\sigma _{\tilde \phi_n} ^2 = \frac{\sigma _{{\omega_{\phi}}} ^2 A_{\phi_n}}{\sigma _{{\omega_{\phi}}} ^2 {\eta \beta^{\mathrm{R}}} + {A_{\phi_n}}}$, where $A_{\phi_n} = \frac{  \Delta t \sigma_s^2 \sigma _R^2 }{{\Delta T {P^{\max}}  \beta_{G,n}^2 L  M_t M_r h(\varphi_n, \sigma^2_{\omega_\varphi}) h(\phi_n, \sigma^2_{\omega_\phi}) {{\sin^2 }}\phi_n } }$.
\begin{Pro}\label{ClosedFormAchievable}
	When the number of IOS elements is sufficiently large, i.e., $L_x \to \infty$ and $L_y \to \infty$, the achievable rate during the joint S\&C stage and the communication-only stage can be respectively approximated as
	\begin{equation}\label{ClosedFormAchievableSC}
			R^{\mathrm{S\&C}}_n \!\approx\! \log_2\left( 1 \!+\! C_0\left(1 - \beta^{\mathrm{R}} \right) h(\varphi_n, \sigma^2_{\omega_\varphi}) h(\phi_n, \sigma^2_{\omega_\phi})  \right)
	\end{equation}
	and
	\begin{equation}\label{ClosedFormAchievalbeC}
			R^{\mathrm{C}}_n \approx \log_2\left( 1 + C_0  h(\varphi_n, \sigma^2_{\tilde \varphi_n}) h(\phi_n, \sigma^2_{\tilde \phi_n})  \right).	
	\end{equation} 
	where $C_0 = \frac{ 4P^{\max} {\beta_{G,n}} \beta_{h}   L M_t }{\sigma_{c}^{2}}$.
\end{Pro}
\begin{proof}
	Please refer to Appendix D.
\end{proof}

\begin{remark}
	It is worth noting that in Propositions \ref{InftyBand} and \ref{ClosedFormAchievable}, unlike the derived squared-power scaling law in \cite{Wu2019IntelligentReflecting}, the SNR at the RSU receive antennas or the communication device scales with the number of IOS elements, i.e., $L$ only. The main reason is that the IOS can capture an inherent aperture gain by collecting more signal power but cannot achieve perfect beamforming gain under the inaccurate target direction. Note that the SNR is greatly affected by the variances of the estimated and measured angles, i.e., $\{\sigma^2_{\omega_\varphi}, \sigma^2_{\omega_\phi}\}$ and $\{\sigma^2_{\tilde \varphi_n}, \sigma^2_{\tilde \phi_n}\}$. Intuitively, the smaller the variances of these angles, the larger the achievable rate.
\end{remark}

In (\ref{ClosedFormAchievableSC}) and (\ref{ClosedFormAchievalbeC}), the accurate angles $\varphi_n$, and  $\phi_n$ cannot be obtained in advance, which can be practically assumed to be $\varphi_{n|n-1}$ and $\phi_{n|n-1}$ since the angle errors are relatively small, which have a negligible effect on the value of $h(\varphi_n, \sigma^2_{\tilde \varphi_n})$ and $h(\phi_n, \sigma^2_{\tilde \phi_n})$. Then, the average rate can be rewritten as a function about $\eta$ and $\beta^{\mathrm{R}}$, i.e., 
\begin{align}\label{AchievableRateReal}
	&R(\eta ,\beta^{\mathrm{R}} ) \\
	\! =&  \eta\log_2 \left( 1 \!+ C_0\left(1 \!-\! \beta^{\mathrm{R}} \right)  h(\varphi_{n|n-1}, \sigma^2_{\omega_\varphi})  h(\phi_{n|n-1}, \sigma^2_{\omega_\phi}) \right) \nonumber\\
	&+  \left(1-\eta\right)\log_2\left( 1 + C_0  h(\varphi_{n|n-1}, \sigma^2_{\tilde \varphi_n})  h(\phi_{n|n-1}, \sigma^2_{\tilde \phi_n})  \right).\nonumber	
\end{align}

Based on the derived expression of the approximate achievable rate in (\ref{AchievableRateReal}), the optimal solution of (P2) can be obtained via a two-dimension search over $\eta$ and $\beta^{\mathrm{R}}$ within $[0,1]$. To gain more useful insights, performance analysis is further provided to illustrate the benefits between S\&C in the next section.

\subsection{Sensing Effectiveness Analysis}
\label{PerformanceAnalysis}
In this subsection, some practical approximations are adopted to draw useful insights. In $h(x,y)$ (defined in Proposition \ref{InftyBand}), it is easy to verify that the term $ {{e^{ - \frac{{{{2\left( {i\pi } \right)}^2}}}{{y}}}} + {e^{ - \frac{{{{2( {(i + 1)\pi } - \varphi_{n|n-1} )}^2}}}{{{y}}}}}}$ is negligible when $i \ge 1$ or $i \le -1$, and thus, we have
\begin{equation}\label{ApproximationEquation}
	h(x,y) \approx \frac{1}{{\sqrt {2\pi{y} \sin^2({x})}  }} \left({1 + {e^{ - \frac{{{{2\left( {\pi  - {x} } \right)}^2}}}{{y}}}}}  \right) \buildrel \Delta \over= \tilde h(x, y) .
\end{equation}
This approximation is practical since the variance of the estimated angle is generally much less than $2\pi$. As such, the rate during the S\&C and communication-only stages can be respectively expressed as
\begin{equation}
	\tilde R^{\mathrm{S\&C}}_n \!=\! \log_2\!\left( 1 \!+ C_0\left(1 \!- \beta^{\mathrm{R}} \right) \tilde h( \varphi_{n|n-1}, \sigma^2_{\omega_\varphi}) \tilde h(\phi_{n|n-1}, \sigma^2_{\omega_\phi})  \right)
\end{equation}
and
\begin{equation}
		\tilde R^{\mathrm{C}}_n  =  \log_2\left( 1 + C_0  \tilde h(\varphi_{n}, \sigma^2_{\tilde \varphi_n}) \tilde h(\hat \phi_n, \sigma^2_{\tilde \phi_{n}})  \right).	
\end{equation} 
Then, the average rate can be approximately given by $\tilde 	R_n = \eta \tilde R^{\mathrm{S\&C}}_n + \left(1-\eta\right) 	\tilde R^{\mathrm{C}}_n$.	

\begin{thm}\label{VarianceMonotonicity}
	As $\eta \beta^{\mathrm{R}}$ increases, $\tilde R^{\mathrm{C}}_n$ increases monotonically.
\end{thm}
\begin{proof}
	First, it is observed that the monotonicity of the rate $\tilde R^{\mathrm{C}}_n$ and the product value $\eta \beta^{\mathrm{R}}$ can be verified by determining the monotonicity of the function $\tilde h( \varphi_{n|n-1}, \sigma^2_{\omega_\varphi})$ with respect to (w.r.t.) $\eta \beta^{\mathrm{R}}$ (similar for $\tilde h( \phi_{n|n-1}, \sigma^2_{\omega_\phi})$), since $\tilde R^{\mathrm{C}}_n$ is a monotonically increasing function of $\tilde h( \varphi_{n|n-1}, \sigma^2_{\omega_\varphi})$. Note that, with any given $\varphi_{n}$, $\tilde h(\varphi_{n}, \sigma^2_{\tilde \varphi_n}) = \frac{g(\eta \beta^{\mathrm{R}})}{{\sqrt {2\pi \sin^2({\varphi_{n}}) }  }}  \left({1 + {e^{ - {{{{2\left( {\pi  - {\varphi_{n}} } \right)}^2}}} g(\eta \beta^{\mathrm{R}}) }}}  \right)
	 \buildrel \Delta \over=  C_1 f(C_2 g(\eta \beta^{\mathrm{R}}))$, where $g(\eta \beta^{\mathrm{R}}) = \frac{\sigma _{{\omega_{\varphi}}} ^2 {\eta \beta^{\mathrm{R}}} + {A_{\varphi_n}}}{\sigma _{{\omega_{\varphi}}} ^2 A_{\varphi_n}}$, $f(z) = \sqrt z \left( {1 + {e^{ - z}}} \right)$, $C_1 = \frac{1}{\sqrt{4 \pi \sin^2(\varphi_{n}) (\pi - \varphi_{n})^2 }} > 0$, and $C_2 = 2(\pi - \varphi_{n})^2 > 0$. Then, we only need to verify $f(z) $ is a monotonically increasing function of $z$ since $\tilde h(\varphi_{n}, \sigma^2_{\tilde \varphi_n})$ and $g(\eta \beta^{\mathrm{R}})$ are affine transformation of $f(z)$ and $\eta \beta^{\mathrm{R}}$, respectively. The derivative of $f(z)$ is expressed as
	%\begin{equation}
		\begin{align}\label{SensingPowerRelationship}
			f'(z) & \!=\! \frac{1}{{2\sqrt z }}\left( {1 + {e^{ - z}}} \right) - \sqrt z {e^{ - z}} \!=\! \frac{1}{{\sqrt z }}\left( {\frac{1}{2} + \left( {\frac{1}{2} - z} \right){e^{ - z}}} \right) \nonumber \\
			&= \frac{1}{{\sqrt z }}\left( {\frac{1}{2} + \frac{1}{{{e^{\frac{1}{2}}}}}\left( {\frac{1}{2} - z} \right){e^{\frac{1}{2} - z}}} \right) \nonumber \\
			&\stackrel{(b)}{\geqslant} \frac{1}{{\sqrt z }}\left( {\frac{1}{2} - \frac{1}{{{e^{\frac{1}{2}}}}}\frac{1}{e}} \right) > 0,
		\end{align}	
	%\end{equation}
	where ($b$) holds due to the inequality $ze^z \ge -\frac{1}{e}$  \cite{corless1996lambertw}. This thus  completes the proof.
\end{proof}
 
\begin{remark}\label{TradeoffFundamental}
	According to Lemma \ref{VarianceMonotonicity}, $R^{\mathrm{S\&C}} \le R^{\mathrm{C}}$ since  $\sigma^2_{\tilde{\varphi_n}} = \frac{\sigma _{\omega_ {\varphi}} ^2 A_{\varphi_n}}{\sigma _{\omega_ {\varphi}} ^2 {\eta \beta^{\mathrm{R}}} + {A_{\varphi_n}}} \ge \sigma _{\omega_ {\varphi}} ^2$. Moreover, when the reflected power $\beta^{\mathrm{R}}$ and the sensing time ratio $\eta$ increases, the rate during the communication-only stage can be improved due to the decreased angle variances $\sigma^2 _{\hat \varphi _n}$ and $\sigma^2 _{\hat \phi _n}$. However, higher sensing power inevitably decreases the rate during the joint S\&C stage, and longer sensing time reduces the duration of the communication-only stage, which may unfavourably reduce the total achievable rate. Thus, it requires an effective balance for $\beta^{\mathrm{R}}$ and $\eta$ to improve the sensing-assisted communication performance. 
\end{remark}

Remark \ref{TradeoffFundamental} unveils a fundamental trade-off between S\&C  performance. Furthermore, a condition that specifies whether it is necessary and sufficient for the IOS to reflect signals for sensing is derived in the following.
For $x \in [\Delta ,\pi - \Delta]$ with ${ \frac{{2{{ \Delta }^2}}}{y}} \gg 1$, we have ${e^{ - \frac{{{{2( {\pi  - x} )}^2}}}{y}}} \ll 1$, since ${{{2{{( \pi - \phi_n )}^2}}}} \gg \sigma^2_{\omega_\phi}$ holds under general parameter setups. In this case, $\tilde R_n$ can be further approximated as
\begin{align}\label{AchievableRateApproximate}
	\hat R_n =& \eta {\log _2}\left( 1 + C_1\left( {1 - \beta^{\mathrm{R}} }\right)  \right)+ ( {1 - \eta } ) \\
	&\times {\log _2}\!\left( 1 \!+ C_1 \sqrt \frac{(\sigma _{{\omega_{\varphi}}} ^2 {\eta \beta^{\mathrm{R}}} \!+ {A_{\varphi_n}})(\sigma _{{\omega_{\phi}}} ^2 {\eta \beta^{\mathrm{R}}} + {A_{\phi_n}})}{  A_{\varphi_n} A_{\phi_n}}   \right) \!, \nonumber
\end{align}
where $C_1 = \frac{ 2 {P^{\max}\beta_{G,n}} \beta_{h}    L M_t }{\pi  \sin(  \varphi_{n|n-1}) \sin(  \phi_{n|n-1}) \sigma_{\omega_{\varphi}} \sigma_{\omega_{\phi}} \sigma_{c}^{2}}$.
Then, this approximate achievable rate is further analyzed to gain more useful insights for system design. In this case, (P2) is simplified to (P3), i.e.,
\begin{alignat}{2}
	\label{P3}
	(\rm{P3}): \quad & \begin{array}{*{20}{c}}
		\mathop {\max }\limits_{\eta ,\beta^{\mathrm{R}}} \quad  \hat R_n
	\end{array} & \\ 
	\mbox{s.t.}\quad
	& \beta^{\mathrm{R}}\in [0,1], \eta \in [0,1] . & \tag{\ref{P3}a}
\end{alignat}
\begin{theorem}\label{SensingCondition}
	At the optimal solution of (P3),  $\eta^* > 0$ if and only if $\frac{{\sigma _{{\omega _\phi }}^2}}{{{A_{\phi_n} }}} + \frac{{\sigma _{{\omega _\varphi }}^2}}{{{A_{\varphi_n} }}} \!>\! 2$.
\end{theorem}
\begin{proof}
	Please refer to Appendix E.
\end{proof}

Theorem \ref{SensingCondition} provides a sufficient and necessary condition on whether the joint S\&C stage is needed. In practice, it tends to be satisfied when the accuracy of estimated angles $\{\varphi_{n|n-1}, \phi_{n|n-1}\}$ is lower or the accuracy of measured angles $\{\hat \varphi_{n}, \hat \phi_{n}\}$ is higher. Intuitively, if $\frac{{\sigma _{{\omega _\phi }}^2}}{{{A_{\phi_n} }}} + \frac{{\sigma _{{\omega _\varphi }}^2}}{{{A_{\varphi_n} }}} \le 2$, the communication performance gain brought by sensing cannot compensate for the performance loss caused by the time and power consumption for sensing during the joint S\&C stage. Thus, when the condition in Theorem \ref{SensingCondition} is not satisfied, $\hat R_n$ tends to be higher if the joint S\&C stage is removed to increase the duration of the communication-only stage. During every frame, it can be dynamically determined whether sensing is needed to improve communication performance based on Theorem \ref{SensingCondition}.

\section{Multi-vehicle Sensing-assisted Communication}
\label{MultiVehicle}
% %%%%%%%%%%%%%%%%%%%%%%%%%%%%%%%%%%%%%%%%%%%%%%%%%%%%%%%%%%%
%% different K eta_k
%%%%%%%%%%%%%%%%%
In this section, we first derive the close-form expression of the communication for multi-vehicle cases and then propose an efficient AO algorithm to solve problem (P1). Then, the interference-limited and noise-limited cases are analyzed to provide a performance bound and a low-complexity algorithm for the considered systems.

\subsection{Closed-form Expression of Achievable Rate and Proposed Solution}
\label{MultiVehicleClosedForm}
Based on the analysis under single-vehicle cases in Section \ref{SingleVehicle}, the closed-form expression of the approximate achievable rate for multi-vehicle cases is derived. First, the interference between different vehicles is analyzed. During the joint S\&C stage, the interference of signal $s_j$ to vehicle $k$ can be given by
\begin{align}\label{Interference_betweenKJ}
	&\mathbb{E}_{\varphi, \phi}\left[ \left|{\bm{h}}^T_k{\bm{\Theta}}^{\mathrm{S\&C,T}}_{k,n} {\bm{H}}^{\mathrm{DL}}_{k,n} {\bm{f}}_{j,n}^{\mathrm{C}}\right|^{2}\right]  \nonumber \\
	 \stackrel{(c)}=&{ P^{\mathrm{S\&C}}_j  {\beta_{G,n}} \beta_{h,k}   L M_t } \mathbb{E}_{\varphi}\left[F_{M_t}\left( \cos(\varphi_{k,n}) - \cos(\varphi_{j,n|n-1})\right)  \right] \nonumber \\
	& \times \mathbb{E}_{\varphi}\left[ F_{L_x}\left( \Delta \cos \varphi_{k,n}\right) \right]  \mathbb{E}_{\phi}\left[F_{L_y}\left( \Delta \cos \phi_{k,n}\right)  \right] ,
\end{align}
where ($c$) holds since the angles $\phi_{k,n}$, $ \varphi_{k,n}$ and $\varphi_{j,n|n-1}$ are independent with each other.
According to Proposition \ref{ClosedFormAchievable}, we have $\mathbb{E}_{\varphi}\left[ F_{L_x}\left( \Delta \cos \varphi_{k,n}\right) \right]  \mathbb{E}_{\phi}\left[F_{L_y}\left( \Delta \cos \phi_{k,n}\right)  \right] = 4h(\varphi_{k,n}, \sigma^2_{\omega_{\varphi}}) h(\phi_{k,n}, \sigma^2_{\omega_{\phi}})$. 
\begin{thm}\label{InterferenceThm}
	The term $\mathbb{E}_{\varphi_{j,n|n-1}}[F_{M_t}( \cos(\varphi_{k,n}) - \cos(\varphi_{j,n|n-1}))]$ in (\ref{Interference_betweenKJ}) can be approximated by 
	\begin{equation}\label{InterferenceTerm}
		\begin{aligned}
		&\mathbb{E}_{\varphi_{j,n|n-1}}\left[F_{M_t}\left( \cos(\varphi_{k,n}) - \cos(\varphi_{j,n|n-1})\right)  \right] \\
			=& I(\varphi_{k,n},  \varphi_{j,n|n-1}, \sigma^2_{\omega_{\varphi}}),
		\end{aligned}
	\end{equation} 
where $I(x, y, z) = \frac{2}{{\sqrt {2\pi {z} \sin^2({x})} }} \sum\limits_{i = -\infty}^{\infty} \bigg( {e^{ - \frac{{{{\left( {2i\pi} + x - y \right)}^2}}}{{2{z}}}}} + {e^{ - \frac{{{{\left( {2(i + 1)\pi  } - x - y\right)}^2}}}{{2{z}}}}} \bigg)$.
\end{thm}
\begin{proof}
	Please refer to Appendix F. 
\end{proof}
	
In Lemma \ref{InterferenceThm}, the interference term $I(x, y, z)$ can be simplified in a similar way as given in Section \ref{PerformanceAnalysis}, i.e., $I(x, y, z) \approx \frac{1}{{\sqrt {2\pi {z} \sin^2({x})} }}\left( {{e^{ - \frac{{{{\left( x - y \right)}^2}}}{{2{z}}}}} + {e^{ - \frac{{{{\left( {2\pi  } - x - y\right)}^2}}}{{2{z}}}}}} \right) \buildrel \Delta \over= \tilde I(x, y, z)$. Similar to (\ref{ClosedFormAchievableSC}) and (\ref{ClosedFormAchievalbeC}), the angle $\varphi_{k,n}$ in (\ref{InterferenceTerm}) can be practically approximated by $\varphi_{k,n|n-1}$. Interestingly, if $\varphi_{j,n} = -\varphi_{k,n}$, the approximate interference between vehicles $j$ and $k$ is equal to each other, i.e., $\tilde I(\varphi_{j,n}, \varphi_{k,n|n-1},  \sigma^2_{\omega_{\varphi}}) = \tilde I(\varphi_{k,n},  \varphi_{j,n|n-1}, \sigma^2_{\omega_{\varphi}})$.
\begin{Pro}\label{ClosedFormAchievableMultiX}
	When the number of IOS elements is sufficiently large, i.e., $L_x \to \infty$ and $L_y \to \infty$, the achievable rate during the joint S\&C stage and the communication-only stage can be respectively approximated as
	\begin{align}\label{ClosedFormAchievableMultiSC}
		&\tilde R^{\mathrm{S\&C}}_{k,n} = {\log _2}\bigg( 1 + \\ 
		&\frac{{4\left( {1 - {\beta_k ^{\mathrm{R}}}} \right)p_k^{\mathrm{S\&C}}{M_t}}}{{4\left( {1 - {\beta_k ^{\mathrm{R}}}} \right)\sum\nolimits_{j \ne k}^K {p_j^{\mathrm{S\&C}}\tilde I(\varphi_{k,n|n-1}, \varphi_{j,n|n-1}, \sigma^2_{\omega_{\varphi}}) }  + \sigma^{\mathrm{S\&C}}_{k,n}}} \Bigg) \nonumber
	\end{align}
	and
	\begin{align}\label{ClosedFormAchievalbeMultiC}
		&\tilde R^{\mathrm{C}}_{k,n}  \\
		=& {\log _2}\!\left(\! {1 \!+\! \frac{{4 p_k^{\mathrm{C}}{M_t}}}{{4\sum\nolimits_{j \ne k}^K {p_j^{\mathrm{C}} \tilde  I(\varphi_{k,n|n-1},  \varphi_{j,n|n-1}, \sigma^2_{{\tilde \varphi_{j,n}}}) } \! +\! \sigma^{\mathrm{C}}_{k,n} }}} \!\right)\!. \nonumber	
	\end{align} 
where $\sigma^{\mathrm{S\&C}}_{k,n} = \frac{{\sigma _c^2}}{{{\beta _{G,k,n}}{\beta _{h,k}}L \tilde h( \varphi_{k,n|n-1}, \sigma^2_{\omega_{\varphi}}) \tilde h( \phi_{k,n|n-1}, \sigma^2_{\omega _{\phi}})}}$, $\sigma^{\mathrm{C}}_{k,n} = \frac{{\sigma _c^2}}{{{\beta _{G,k,n}}{\beta _{h,k}}L \tilde h( \varphi_{k,n|n-1}, \sigma^2_{\tilde \varphi_{k,n}}) \tilde h( \phi_{k,n|n-1}, \sigma^2_{\tilde \phi_n})}}$, $\sigma _{\tilde \varphi_{k,n}} ^2 = \frac{\sigma _{{\omega_{\varphi}}} ^2 A_{\varphi_{k,n}}}{\sigma _{{\omega_{\varphi}}} ^2 {p^{\mathrm{S\&C}}_k \eta \beta_k^{\mathrm{R}}} + {A_{\varphi_{k,n}}}}$, $\sigma _{\tilde \phi_{k,n}} ^2 \!\!=\!\! \frac{\sigma _{{\omega_{\phi}}} ^2 A_{\phi_{k,n}}}{\sigma _{{\omega_{\phi}}} ^2 {p^{\mathrm{S\&C}}_k \eta \beta_k^{\mathrm{R}}} + {A_{\phi_{k,n}}}}$, $A_{\varphi_{k,n}} \!=\! \frac{  \Delta t \sigma_s^2 \sigma _R^2 }{{\Delta T \beta_{G,k,n}^2 L  M_t M_r h(\varphi_{k,n}, \sigma^2_{\omega_{\varphi}}) h(\phi_{k,n}, \sigma^2_{\omega_{\phi}}) {{\sin^2 }}\varphi_{k,n} } }$, and $A_{\phi_{k,n}} \!=\! \frac{  \Delta t \sigma_s^2 \sigma _R^2 }{{\Delta T   \beta_{G,k,n}^2 L  M_t M_r h(\varphi_{k,n}, \sigma^2_{\omega_{\varphi}}) h(\phi_{k,n}, \sigma^2_{\omega_{\phi}}) {{\sin^2 }}\phi_{k,n} } }$.
\end{Pro}
\begin{proof}
	Combing the results in Proposition \ref{ClosedFormAchievable} and Lemma \ref{InterferenceThm}, the approximate achievable rate of communication device $k$ can be obtained by plugging the terms in (\ref{ClosedFormAchievableSC}) and (\ref{InterferenceTerm}) into (\ref{AchievablerateSC}) and (\ref{AchievablerateC}), thereby obtaining the expressions in (\ref{ClosedFormAchievableMultiSC}) and (\ref{ClosedFormAchievalbeMultiC}).
\end{proof}

Proposition \ref{ClosedFormAchievableMultiX} provides a closed-form approximate achievable rate, based on which, (P1) can be rewritten as 
\begin{alignat}{2}
	\label{P4}
	(\rm{P4}): \quad & \begin{array}{*{20}{c}}
		\mathop {\max }\limits_{{\bm p}^{\mathrm{S\&C}}_k, {\bm p}^{\mathrm{C}}_k, {\bm \beta}^{\mathrm{R}}, \eta} \quad   R
	\end{array} & \\ 
	\mbox{s.t.}\quad
	& (\ref{P1}c)-(\ref{P1}f), \nonumber \\
	& \eta \tilde R^{\mathrm{S\&C}}_{k,n} +  \left( {1 - \eta } \right) \tilde R^{\mathrm{C}}_{k,n} \ge R , \forall k, & \tag{\ref{P4}a} \\
	& \beta _k^{\mathrm{R}} \in [0, 1], \forall k, \tag{\ref{P4}b} 
\end{alignat}
where ${\bm \beta}^{\mathrm{R}} = [{\beta}_1^{\mathrm{R}}, \cdots, {\beta}_K^{\mathrm{R}}]$. However, it is still challenging to solve (P1) due to the non-convex constraints and the closely coupled variables in (\ref{P4}a). According to Lemma \ref{VarianceMonotonicity}, the rate during the communication-only stage depends mainly on the sensing performance, i.e., the variances $\sigma^2_{\tilde \varphi_{k,n}}$ and $\sigma^2_{\tilde \phi_{k,n}}$. Thus, a low-complexity power allocation method during the joint S\&C stage is adopted, where the sensing performance of different vehicles is designed to be equal, thereby simplifying the solution to this problem. Specifically, let $\frac{A_{\varphi_{k,n}}}{p^{\mathrm{S\&C}}_k} = \frac{A_{\varphi_{j,n}}}{p^{\mathrm{S\&C}}_j}$ and $\beta^{\mathrm{R}}_k = \beta^{\mathrm{R}}_j$, $\forall k,j \in {\cal{K}}$, then $\sigma^2_{\tilde \varphi_{k,n}} = \sigma^2_{\tilde \varphi_{j,n}}$ and $\sigma^2_{\tilde \phi_{k,n}} = \sigma^2_{\tilde \phi_{j,n}}$. In this case, an AO algorithm is proposed to solve (P4) by partitioning the remaining variables into three blocks and iteratively optimizing them, i.e., transmit power ${\bm p}^{\mathrm{C}}$, time splitting ratio $\eta$, and reflecting power ratio ${\beta}^{\mathrm{R}}_k$. 

\subsubsection{Transmit Power}
With any given $\eta$ and ${\bm \beta}^{\mathrm{R}}$, the successive convex approximation technique is adopted to handle (\ref{P4}a). A lower bound of $\tilde R^{\mathrm{C}}_{k,n}$ can be given as follows:
\begin{align}
	& \tilde R^{\mathrm{C}}_{k,n} \!\ge\! -  {\log_2}\!\left(\!4\sum\nolimits_{j \ne k}^K {p_j^{\mathrm{C}, (r)}\tilde I(\varphi_{k,n|n-1}, \varphi_{j,n|n-1}, \sigma^2_{{\tilde \varphi_j}}) } \!+\! \sigma^{\mathrm{C}}_{k,n}\! \!\right) \nonumber \\
	&  + \!{\log _2} \!\bigg( \! 4\sum\nolimits_{j \ne k}^K {p_j^{\mathrm{C}}\tilde I(\varphi_{k,n|n-1}, \varphi_{j,n|n-1}, \sigma^2_{{\tilde \varphi_j}}) }   \!+ \! \sigma^{\mathrm{C}}_{k,n} \! + {{4p_k^{\mathrm{C}}{M_t}}}  \!\bigg)  \nonumber \\
	&- \sum\nolimits_{j \ne k}^K D^{\mathrm{C}}_{k,j} (p_j^{\mathrm{C}}-p_j^{\mathrm{C}, (r)}) \buildrel \Delta \over= \underline R^{\mathrm{C}}_{k,n},	
\end{align}
where $D^{\mathrm{C}}_{k,j} = \frac{4\tilde I(\varphi_{k,n|n-1}, \varphi_{j,n|n-1}, \sigma^2_{{\tilde \varphi_j}}) \ln 2}{4\sum\nolimits_{j \ne k}^K{p_j^{\mathrm{C}, (r)}\tilde I(\varphi_{k,n|n-1}, \varphi_{j,n|n-1}, \sigma^2_{{\tilde \varphi_j}}) } + \sigma^{\mathrm{C}}_{k,n}}$. Then, with any given local point $\{p_k^{\mathrm{C}, (r)}\}$, the transmit power can be updated by solving (P4.1) as follows.
\begin{alignat}{2}
	\label{P4.1}
	(\rm{P4.1}): \quad & \begin{array}{*{20}{c}}
		\mathop {\max }\limits_{{\bm p}^{\mathrm{C}}} \quad   R
	\end{array} & \\ 
	\mbox{s.t.}\quad
	& (\ref{P1}c)-(\ref{P1}f), (\ref{P4}b), \nonumber \\
	& \eta \tilde R^{\mathrm{S\&C}}_{k,n} +  \left( {1 - \eta } \right) \underline R^{\mathrm{C}}_{k,n} \ge R , \forall k. & \tag{\ref{P4.1}a} 
\end{alignat}
(P4.1) is a convex optimization problem about ${\bm{p}}^{\mathrm{C}}$, which can be optimally solved by standard CVX tools.
\subsubsection{Sensing Time}
With any given ${\bm p}^{\mathrm{C}}$, and ${\bm \beta}^{\mathrm{R}}$, the optimal solution of $\eta$ can be obtained by one-dimension search within $[0, 1]$.
\subsubsection{Reflecting Power}
Similarly, with any given $\eta$ and ${\bm p}^{\mathrm{C}}$, the optimal solution of $\beta _k^{\mathrm{R}}$ can be obtained by one-dimension search within $[0, 1]$.

The objective function of (P4) is non-increasing over each iteration during applying the AO method and the objective of (P4) is upper bounded due to the limited transmit power $P_{\max}$ \cite{Wu2018JointTrajectory}. Thus, by alternately solving three block variables, it is guaranteed to converge. To gain more design insights, some typical cases of the considered system are further investigated in the next section.

The complexity of the proposed AO algorithm for general multi-vehicle cases can be analyzed as follows. In the inner layer, the main complexity is caused by the computation of ${\bm{p}}^C$, i.e., ${\cal{O}}\left((K + 1)^{3.5}\right)$ \cite{zhang2019securing}, where $K + 1$ is the number of variables in problem (\rm{P4.1}). Thus, the complexity of the proposed algorithm for the inner layer is ${\cal{O}}\left((K + 1)^{3.5}\right)$, and its overall complexity depends on the number of iterations required for reaching convergence in the outer layer.

\subsection{Interference-Limited and Noise-Limited Cases}
\label{LowCOmplexity}
In this subsection, two typical scenarios for V2X networks are investigated to draw more useful insights, including interference-limited cases and noise-limited cases. 

\subsubsection{Interference Limited (IL)}
When some vehicles are relatively close or the transmit power of the RSU is relatively large, co-channel interference is the main bottleneck of communication performance. In this case, it is practically assumed that the noise power is negligible. To draw useful insights, we present the analysis under a high transmit power region, where the sensing results are assumed to be relatively accurate due to the high SNR of echo signals, i.e., $ \sigma^2_{z_{\varphi_{k,n}}} \to 0$. In this case, the optimal reflecting power ratio $\beta^{{\mathrm{R}}*}_k > 0$ and the SINR can be approximately expressed by the signal-to-interference-ratio (SIR) at the communication device, then the achievable rate during two stages can be respectively approximated as
\begin{equation}\label{ClosedFormAchievableInterferenceLImited}
	R  ^{\mathrm{IL, S\&C}}_{k,n} = \!{\log _2}\!\left( {1 + \frac{{p_k^{\mathrm{S\&C}}{M_t}}}{{\sum\nolimits_{j \ne k}^K {p_j^{\mathrm{S\&C}}\tilde I(\varphi_{j,n|n-1}, \varphi_{k,n|n-1}, \sigma^2_{\omega_{\varphi}}) }}}} \right)
\end{equation}
and
\begin{equation}\label{ClosedFormAchievableCInterferenceLImited}
	R^{\mathrm{IL,C}}_{k,n} = {\log _2}\left( {1 + \frac{{ p_k^{\mathrm{C}}{M_t}}}{{\sum\nolimits_{j \ne k}^K {p_j^{\mathrm{C}}   F_{M_t}\left( \cos(\hat \varphi_{k,n}) - \cos(\hat \varphi_{j})\right) } }}} \right).	
\end{equation} 
Then, (P4) can be simplified to (P4.2), given by
\begin{alignat}{2}
	\label{P4.2}
	(\rm{P4.2}): \quad & \begin{array}{*{20}{c}}
		\mathop {\max }\limits_{{\bm p}^{\mathrm{S\&C}}_k, {\bm p}^{\mathrm{C}}_k, {\bm \beta}^{\mathrm{R}}_k, \eta} \quad   R
	\end{array} & \\ 
	\mbox{s.t.}\quad
	& (\ref{P1}c)-(\ref{P1}f), (\ref{P4}b), \nonumber \\
	& \eta R  ^{\mathrm{IL, S\&C}}_{k,n} +  \left( {1 - \eta } \right) R^{\mathrm{IL,C}}_{k,n}  \ge R , \forall k. & \tag{\ref{P4.2}a} 
\end{alignat}

\begin{Pro}\label{UpperBound}
	When the maximum transmit power $P^{\max}$ is larger than a saturation point, an upper bound of the rate of (P4.2) can be given by
	\begin{equation}
	\bar R_n \!=\! {\log _2}\!\left(\! {1 \!+\! \frac{M_t}{{\mathop {\min }\limits_k \sum\nolimits_{j \ne k}^K {{F_{M_t}\left(\cos(\hat \varphi_{k,n}), \cos(\hat \varphi_{j,n})\right)}} }}	} \!\right)\!.
\end{equation}
\end{Pro}
\begin{proof}
	Please refer to Appendix G.
\end{proof}

Proposition \ref{UpperBound} illustrates that under high transmit power regions, there exists a performance ceiling regardless of $P^{\max}$. Intuitively, as the transmit power increases, the received signal power and the corresponding interference will increase by the same scale. Thus, the SIR of the communication device keeps constant when the transmit power increases to a certain saturation point. In this case, the upper bound of the communication performance is affected by the minimum interference of the communication devices. An AO method similar to the proposed algorithm in Section \ref{MultiVehicleClosedForm} can be adopted to iteratively optimize the transmit power, the reflecting power ratio, and the time splitting ratio, and the details are omitted for brevity. 

\subsubsection{Noise Limited (NL)}
\label{SectionNLCase}
When the number of antennas $M_t$ increases, the interference between different users decreases, and it can be ignored when the number of antennas $M_t$ is sufficiently large \cite{ngo2015massive}. Also, when the distance between vehicles is relatively large, the potential interference among vehicles is negligible. In this case, the achievable rate during the joint S\&C stage and the communication-only stage can be respectively approximated as
\begin{equation}\label{ClosedFormAchievableMultiVehiclesSC}
	R^{\mathrm{NL,S\&C}}_{k,n} \!\!=\! \log_2\!\left(\! 1 \!+\! C_2 P^{\mathrm{S\&C}}_k (1 \!-\! \beta^{\mathrm{R}} ) h(\varphi_{k,n}, \sigma^2_{\omega_{\varphi}}\!) h(\phi_{k,n}, \sigma^2_{\omega_{\phi}})  \!\right)
\end{equation}
and
\begin{equation}\label{ClosedFormAchievalbeMultiVehiclesC}
	R^{\mathrm{NL,C}}_{k,n} \!= \log_2\!\left(\! 1 \!+ C_2 P^{\mathrm{C}}_k  h(\varphi_{k,n}, \sigma^2_{\tilde \varphi_{k,n}}) h(\phi_{k,n}, \sigma^2_{\tilde \varphi_{k,n}})  \!\right).	
\end{equation} 
where $C_2 = \frac{ 4 {\beta_{G,k,n}} \beta_{h,k}   L M_t }{\sigma_{c}^{2}}$. If $\eta^* = 0$ and $\beta^{{\mathrm{R}}*}_k = 0$, the problem is reduced to a power control problem, which can be solved by the water-filling algorithm \cite{boyd2004convex}. Then, if the obtained transmit power ${\bm{p}}^{\mathrm{S\&C}}$ and ${\bm{p}}^{\mathrm{C}}$ via the water-filling algorithm do not make the condition in Theorem \ref{SensingCondition} satisfied, i.e., $\frac{{\sigma _{{\omega _\phi }}^2}}{{{A_{\phi_{k,n}} }}} + \frac{{\sigma _{{\omega _\varphi }}^2}}{{{A_{\varphi_{k,n}} }}} \le 2, \forall k \in {\cal{K}}$, ${\bm{p}}^{\mathrm{S\&C}}$ and ${\bm{p}}^{\mathrm{C}}$ are the optimal solution under noise-limited cases. Otherwise, (P4) can be rewritten as 
\begin{alignat}{2}
	\label{P4.3}
	(\rm{P4.3}):  & \begin{array}{*{20}{c}}
		\mathop {\max }\limits_{{\bm{p}}^{\mathrm{S\&C}}, {\bm{p}}^{\mathrm{C}}, {\bm{\beta^{\mathrm{R}}}},\eta} \quad   R
	\end{array} & \\ 
	\mbox{s.t.}\quad
	& (\ref{P1}c)-(\ref{P1}f), (\ref{P4}b) \nonumber \\
	& \eta R  ^{\mathrm{NL, S\&C}}_{k,n} +  \left( {1 - \eta } \right) R^{\mathrm{NL,C}}_{k,n}  \ge R , \forall k. & \tag{\ref{P4.3}a} 
\end{alignat}
With any given $\eta$ and ${\bm \beta}^{\mathrm{R}}$, the optimization of transmit power ${\bm p}^{\mathrm{S\&C}}$ and ${\bm p}^{\mathrm{C}}$ is a convex problem, and it can be solved by the interior-point method \cite{boyd2004convex}. Then, the AO algorithm proposed in Section IV-A can be extended to solve (P4.3) by iteratively optimizing the transmit power, the time splitting ratio, and the reflecting power in a similar way.

\section{Simulation Results}
\label{simulationS}

In this section, simulation results are provided for characterizing the performance of the proposed sensing scheme and for gaining insight into the design and implementation of IOS-assisted ISAC systems. The main system parameters are listed in Table \ref{tab1}. The coordinate of the RSU is (0 m, 0 m, 20 m), the vehicle drives from (-100 m, 20 m, 0 m) to (100 m, 20 m, 0 m), and the vehicle's speed is 20 m/s. We compare the proposed schemes with three benchmarks: 
\begin{itemize}
	\item {\bf{Refraction optimization}}: The reflecting phase shifts ${\bm{\Theta}}^{\xi, {\mathrm{R}}}_{n}$ are random, with the optimal $\eta$, $\beta^{\mathrm{R}}_k$, and ${\bm{\Theta}}^{\xi, {\mathrm{T}}}_{k,n}$.
	\item {\bf{Prediction}}: The beamforming and phase shifts are designed only based on the state estimation, i.e., $\eta$ = 0. 
	\item {\bf{Beam Training}}: The RSU transmits multiple training beams sequentially, and then the vehicle sends the index of the beam with the maximum received power to the RSU.  
\end{itemize}
\begin{table}[h]
	\footnotesize
	\centering
	\caption{System Parameters}
	\label{tab1}
	\begin{IEEEeqnarraybox}[\IEEEeqnarraystrutmode\IEEEeqnarraystrutsizeadd{2pt}{1pt}]{v/c/v/c/v/c/v/c/v/c/v/c/v}
		\IEEEeqnarrayrulerow\\
		&\mbox{Parameter}&&\mbox{Value}&&\mbox{Parameter}&&\mbox{Value}&&\mbox{Parameter}&&\mbox{Value}&\\
		\IEEEeqnarraydblrulerow\\
		\IEEEeqnarrayseprow[3pt]\\
		&K&&5&&N&&200&&M_t&&8&\IEEEeqnarraystrutsize{0pt}{0pt}\\
		\IEEEeqnarrayseprow[3pt]\\
		\IEEEeqnarrayrulerow\\
		\IEEEeqnarrayseprow[3pt]\\  
		&M_r&& 8 &&L_x&&80&&L_y&&80&\IEEEeqnarraystrutsize{0pt}{0pt}\\
		\IEEEeqnarrayseprow[3pt]\\
		\IEEEeqnarrayrulerow\\
		\IEEEeqnarrayseprow[3pt]\\  
		&\sigma_R^2&& 10^{-1} &&P^{\max}&&0.1 \; \mbox{W}&&\beta_0&& -30 \; \mbox{dB} &\IEEEeqnarraystrutsize{0pt}{0pt}\\
		\IEEEeqnarrayseprow[3pt]\\
		\IEEEeqnarrayrulerow\\
		\IEEEeqnarrayseprow[3pt]\\
		&\sigma_s^2&&-70 \; \mbox{dBm}&&\sigma^2_{\omega_\phi} \ \sigma^2_{\omega_\varphi}&&0.1 &&\Delta t&&0.1 \;  \mu\mbox{s} &\IEEEeqnarraystrutsize{0pt}{0pt}\\
		\IEEEeqnarrayseprow[3pt]\\
		\IEEEeqnarrayrulerow\\
		\IEEEeqnarrayseprow[3pt]\\
		&\sigma_c^2&& -70 \; \mbox{dBm} &&f_c && 30 \;  \mbox{GHz} && \Delta T && 0.02 \;  \mbox{s} &\IEEEeqnarraystrutsize{0pt}{0pt}\\
		\IEEEeqnarrayseprow[3pt]\\
		\IEEEeqnarrayrulerow
	\end{IEEEeqnarraybox}  
\end{table} 
\begin{figure*}[t]
	\centering
	\setlength{\abovecaptionskip}{0.cm}
	\subfigure[SNR of echo signals in Proposition \ref{InftyBand}.]
	{	
		\label{figure4a}
		\includegraphics[width=5.3cm]{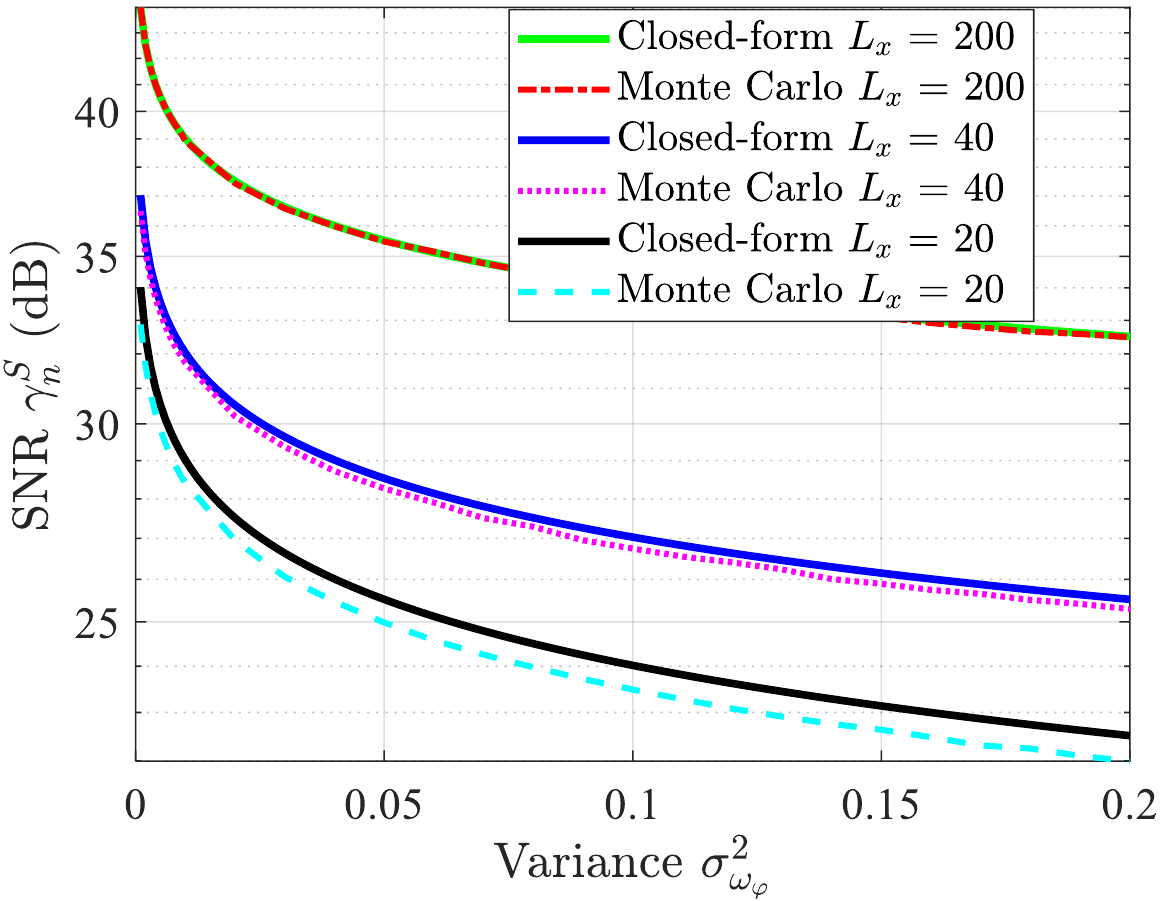}
	}
	\subfigure[Interference term in Lemma \ref{InterferenceThm}.]
	{	
		\label{figure4b}
		\includegraphics[width=5.1cm]{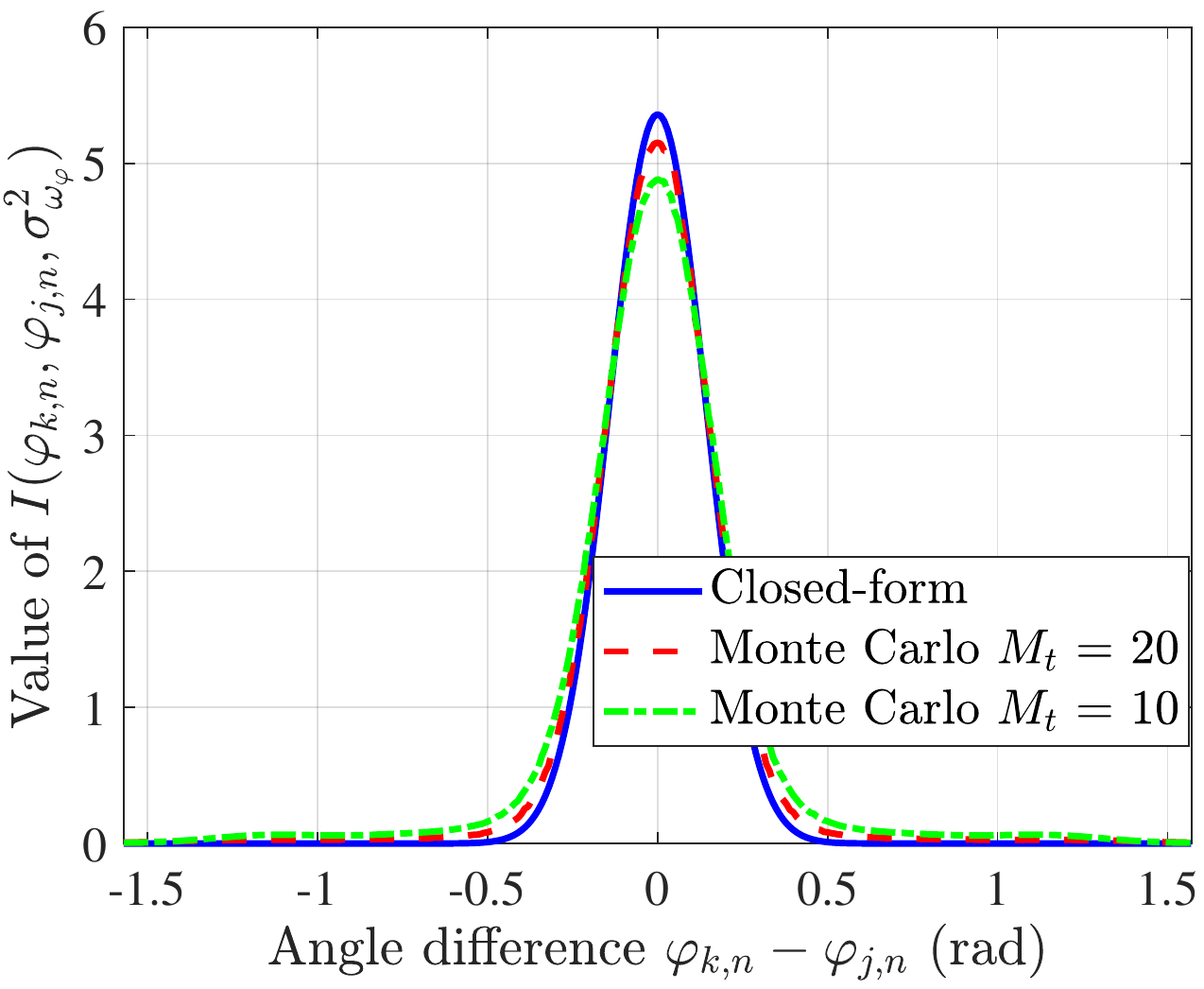}
	}
	\subfigure[Achievable rate in Proposition \ref{ClosedFormAchievable}.]
	{	
		\label{figure4c}
		\includegraphics[width=5.3cm]{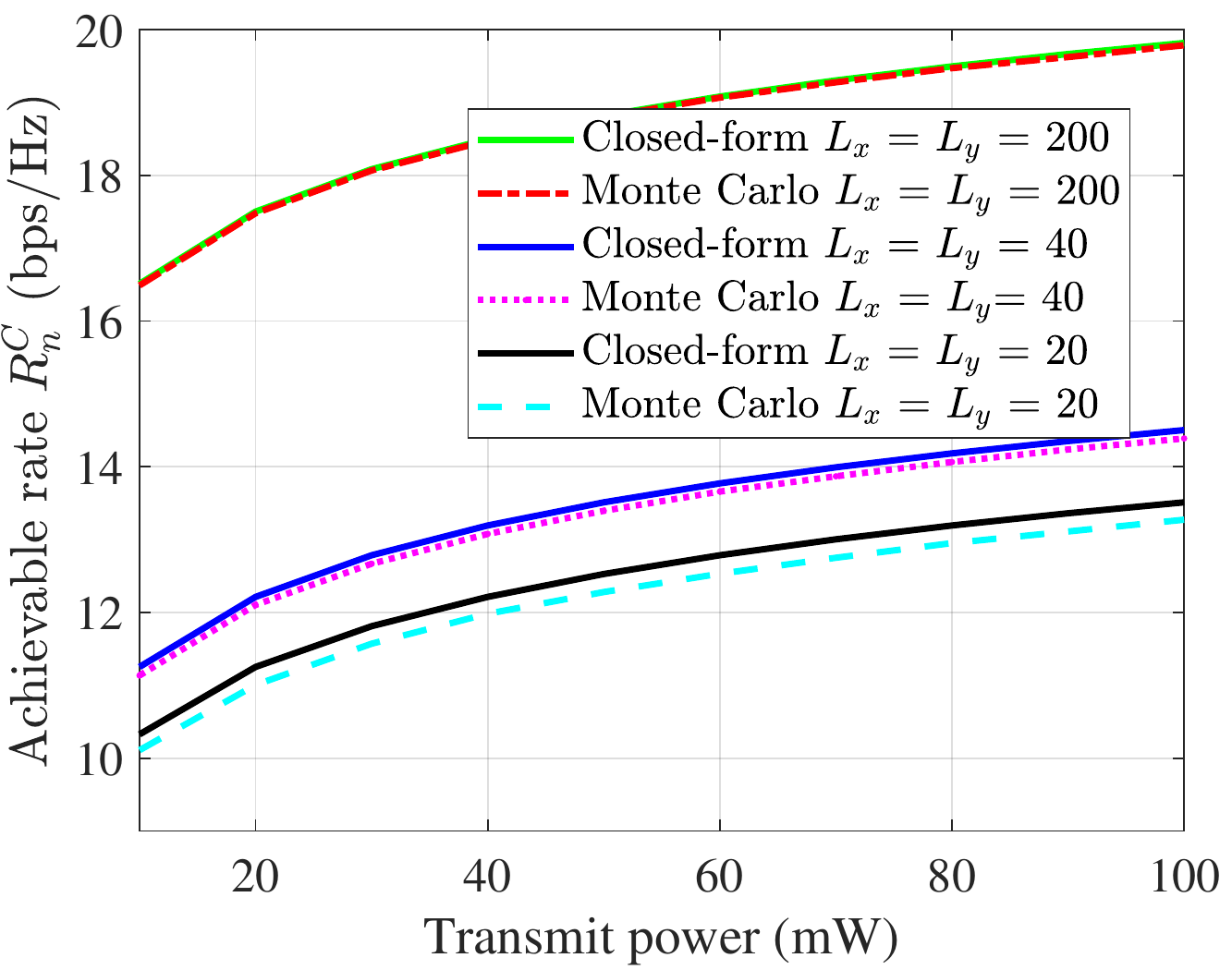}
	}		
	\caption{Evaluation of the closed-form expressions in this work.}
	\label{figure4}
\end{figure*}

\subsection{Evaluation of Closed-form Expression}
Fig.~\ref{figure4} illustrates the effectiveness of the derived closed-form expression in this work, with $L_y = 1$. Specifically, in Fig.~\ref{figure4a}, the SNR $\gamma^{{\mathrm{S}},x}_n$ of Monte Carlo simulation with $L_x = 40$ approaches to the closed-form expression of the SNR derived in Proposition \ref{InftyBand}, which confirms the goodness of this approximation. In Fig.~\ref{figure4b}, the approximate interference under different numbers of transmit antennas is evaluated with $\sigma^2_{\omega_{\varphi}} = 0.02$. It can be seen that the interference decreases as the angle difference increases and the interference is negligible when the angle difference is larger than $0.25$ rad in this case. Fig.~\ref{figure4c}, with $\sigma^2_{\omega_{\varphi}} = 0.02$, the closed-form expression of the approximate achievable rate in Proposition \ref{ClosedFormAchievable} is evaluated under different transmit power, and the Monte Carlo simulation results also confirm the validity of the approximation since the achievable rate of Monte Carlo simulations with $L_x = 40$ approaches to the closed-form expression of the achievable rate derived in Proposition \ref{ClosedFormAchievable}.

\begin{figure}[t]
	\centering
	\includegraphics[width=7cm]{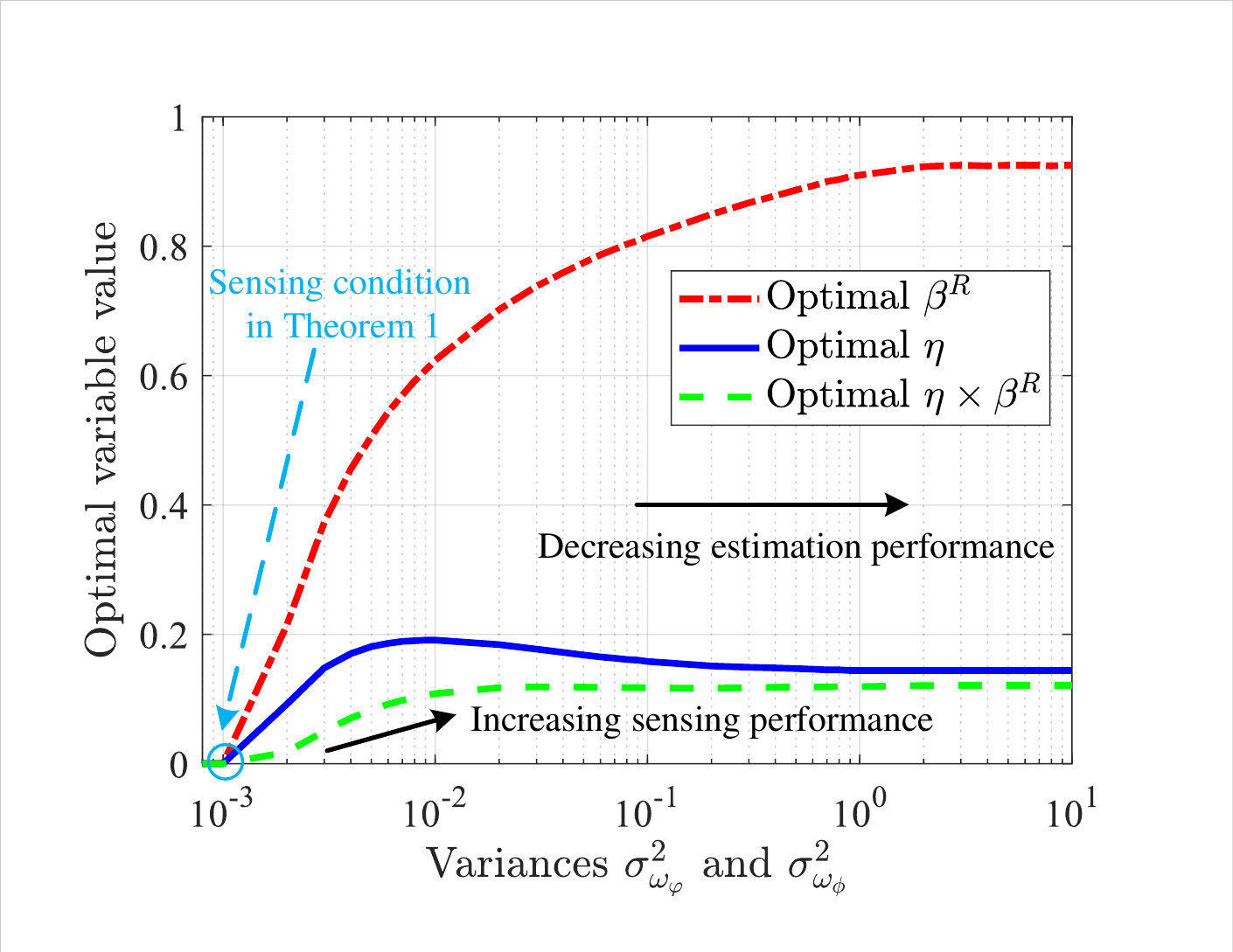}
	\vspace{-2mm}
	\caption{Optimal variable values for different estimation variances.}
	\label{figure6aa}
\end{figure}
\begin{figure}[t]
	\centering
	\includegraphics[width=7cm]{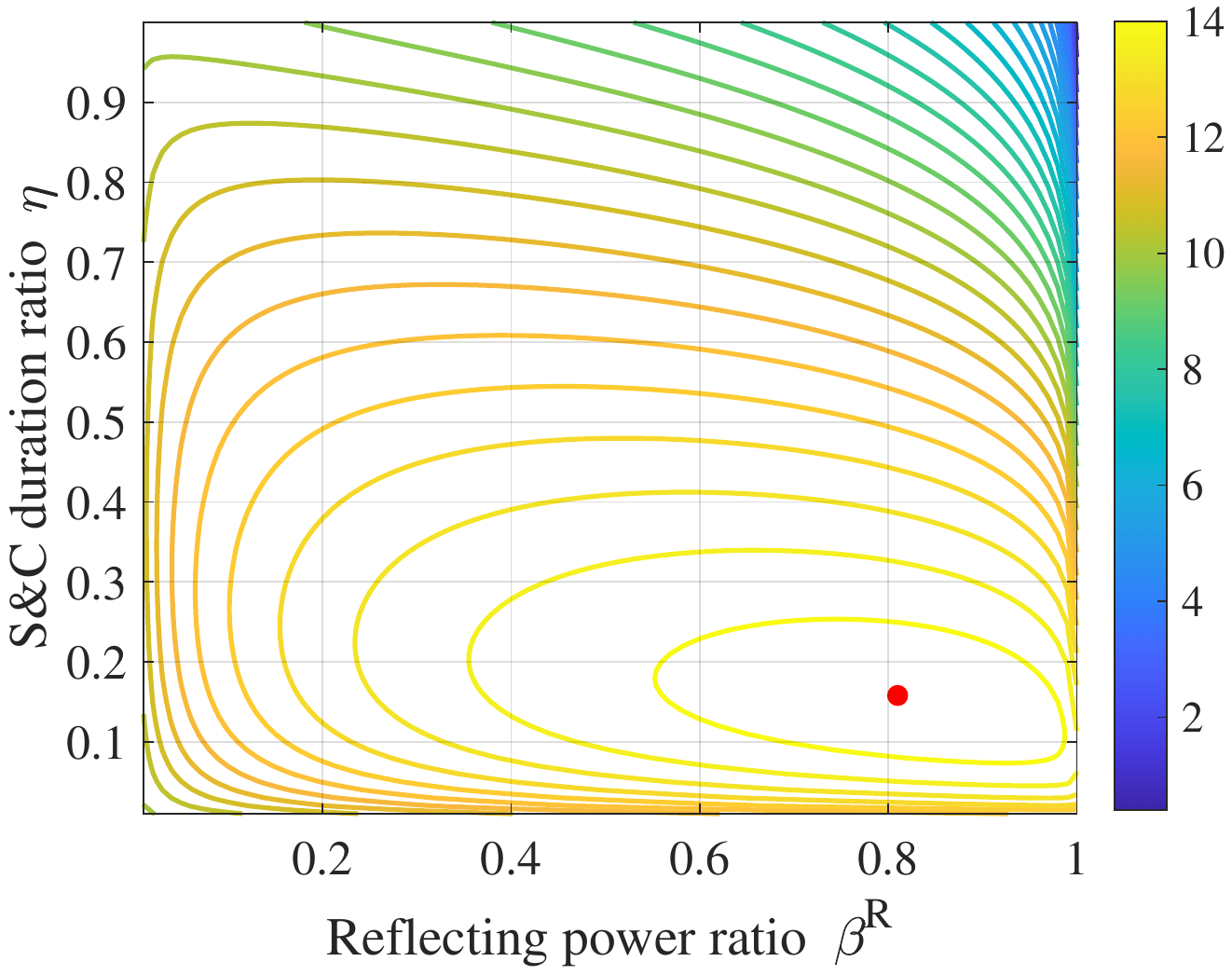}
	\caption{Impact of $\beta^{\mathrm{R}}$ and $\eta$ on the achievable rate.}
	\label{figure6b}
\end{figure}

To show the impact of discrete phase shifts, we quantize the obtained continuous IOS phase shifts to the nearest values, as adopted in reference \cite{Wu2020Beamforming}. Here, let $b$ denote the number of bits, which represents the resolution levels of IOS. Then, the $l$-th discrete phase shift, denoted by $\hat{\nu}_l$, can be derived from $\hat{\nu}_l=\underset{\nu \in \mathcal{B}}{\arg \min }\left|\nu -\theta_l\right|
$
where $\mathcal{B}=\left\{0,2 \pi / 2^b, \ldots, 2 \pi\left(2^b-1\right) / 2^b\right\}$, and $\theta_l$ represents the $l$-th IOS phase shift obtained by solving the resource allocation problem for the proposed sensing-assisted communication scheme. Fig. \ref{figure5} shows the power of echo signals versus estimation variance with $L_y = 1$, where both $b=1$ and $b=3$ for the discrete phase shift scheme are considered at the IOS. It can be observed that a 3-bit phase shifter achieves much better performance than a 1-bit phase shifter. Specifically, with $L_x = 200$, there is a power loss of around $0.3 / 6 \mathrm{~dB}$ from the $3 / 1$-bit phase shifter. Also, the performance loss with $L_x = 200$ is larger than that with $L_x = 20$, which is consistent with the analytical result in \cite{Wu2020Beamforming}. Thus, the performance loss is negligible for the quantization scheme with $b = 3$, and the continuous phase shift designs could provide a simpler way to explore the performance upper bounds of the sensing-assisted communication scheme.

\begin{figure}[t]
	\centering
	\includegraphics[width=7cm]{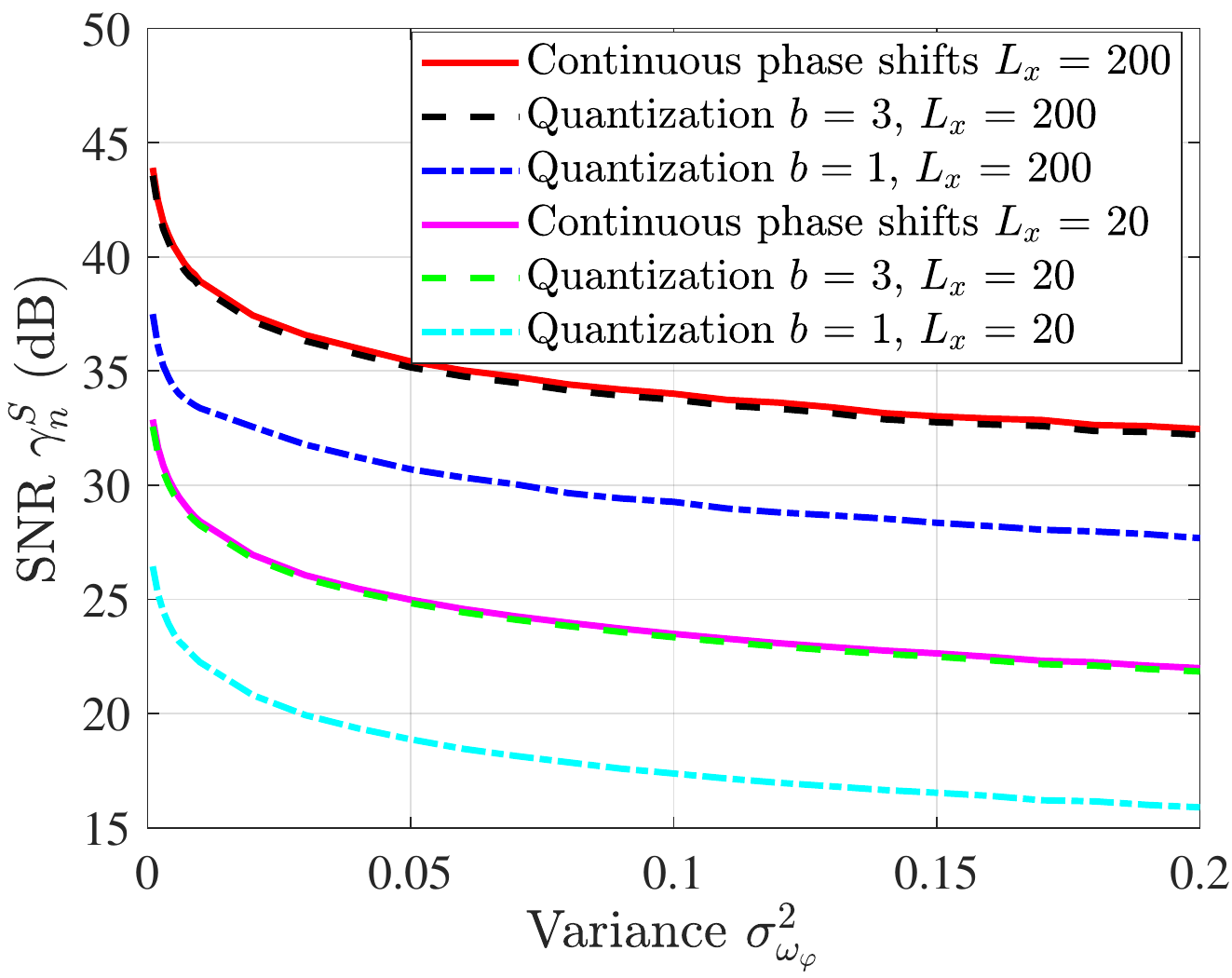}
	\vspace{-2mm}
	\caption{Echo power comparison between continues and quantization scheme.}
	\label{figure5}
\end{figure}

\subsection{Power and Time Allocation}
In Fig.~\ref{figure6aa}, when the variances of estimated angles $\sigma^2_{\omega_\varphi}$ and $\sigma^2_{\omega_\phi}$ are larger than $10^{-3}$, the necessary condition of the existence of the joint S\&C stage is satisfied (i.e., $\eta^* > 0$), which confirms the analysis in Theorem \ref{SensingCondition}. As $\sigma^2_{\omega_\varphi}$ and $\sigma^2_{\omega_\phi}$ increase, the optimal reflecting power $\beta^{{\mathrm{R}}*}$ and the optimal product value $\eta \beta^{\mathrm{R}}$ monotonically increase. Intuitively, at lower estimation accuracy, the sensing results can more efficiently improve the achievable rate during the communication-only stage, and thus the average communication performance can be enhanced by properly allocating the time and power resources. 

Moreover, Fig.~\ref{figure6b} shows the rate of the proposed scheme under different reflecting power $\beta^{\mathrm{R}}$ and time splitting ratio  $\eta$. Specifically, the red dot in Fig.~\ref{figure6b} represents the optimal solution of (P2), i.e., $\eta^* = 0.16$ and $\beta^{{\mathrm{R}}*} = 0.81$. It can be seen that with any given $\beta^{\mathrm{R}}$, the rate tends to first increase and then decreases with the increase of the S\&C stage ratio $\eta$. The main reason is that an appropriate increase in the sensing duration can reduce the variance of parameter measurements and improve the beamforming gain, however, the communication performance gain brought by sensing cannot compensate for the performance loss caused by excessive sensing time and power consumption during the joint S\&C stage, especially for ratio $\eta$ approaching one.

\subsection{Achievable Rate Comparison For Single-Vehicle Cases}

In Fig.~\ref{figure8}, it can be seen that the achievable rate of the prediction and refraction schemes fluctuates greatly since the state estimation error of the vehicle location is relatively large, resulting in a
significant performance loss due to the misalignment of the transmit beam of the RSU and the refracted beam of the IOS. The main reason is that when the number of IOS elements or RSU antennas is large, a slight angular deviation would bring a relatively large beamforming performance loss at the RSU or the IOS. The beam training scheme provides a relatively stable rate, but its communication performance is about 20\% lower than that of the proposed scheme due to the high overhead of the beam pattern training as well as the quantization error of the beam direction. The proposed method provides higher throughput and stable data transmission for the communication device. The main reason is that both the state estimation and measurement results are exploited to improve the communication performance, and the optimal ratios of sensing time and reflecting power are taken to achieve a better balance between S\&C performance, thereby leading to an effective communication improvement benefited from sensing. Compared to the perfect baseline with accurate location information, the performance loss of the proposed scheme mainly originates from the estimation and measurement noises and the extra power consumption used to improve the sensing performance.
\begin{figure}[t]
	\centering
	\includegraphics[width=7cm]{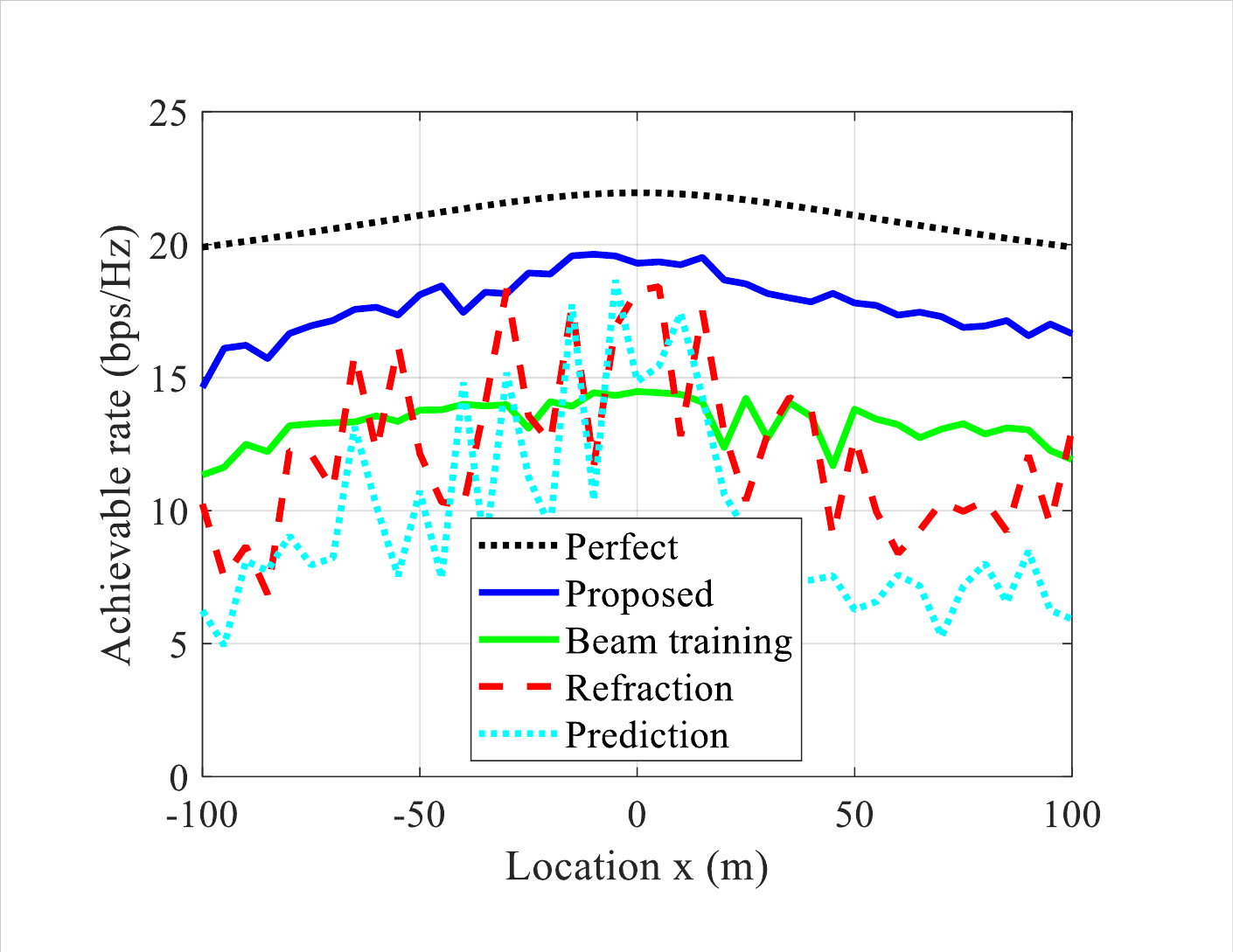}
	\vspace{-2mm}
	\caption{Comparison of the instant achievable rate versus the vehicle location.}
	\label{figure8}
\end{figure}
\begin{figure}[t]
	\centering
	\includegraphics[width=7cm]{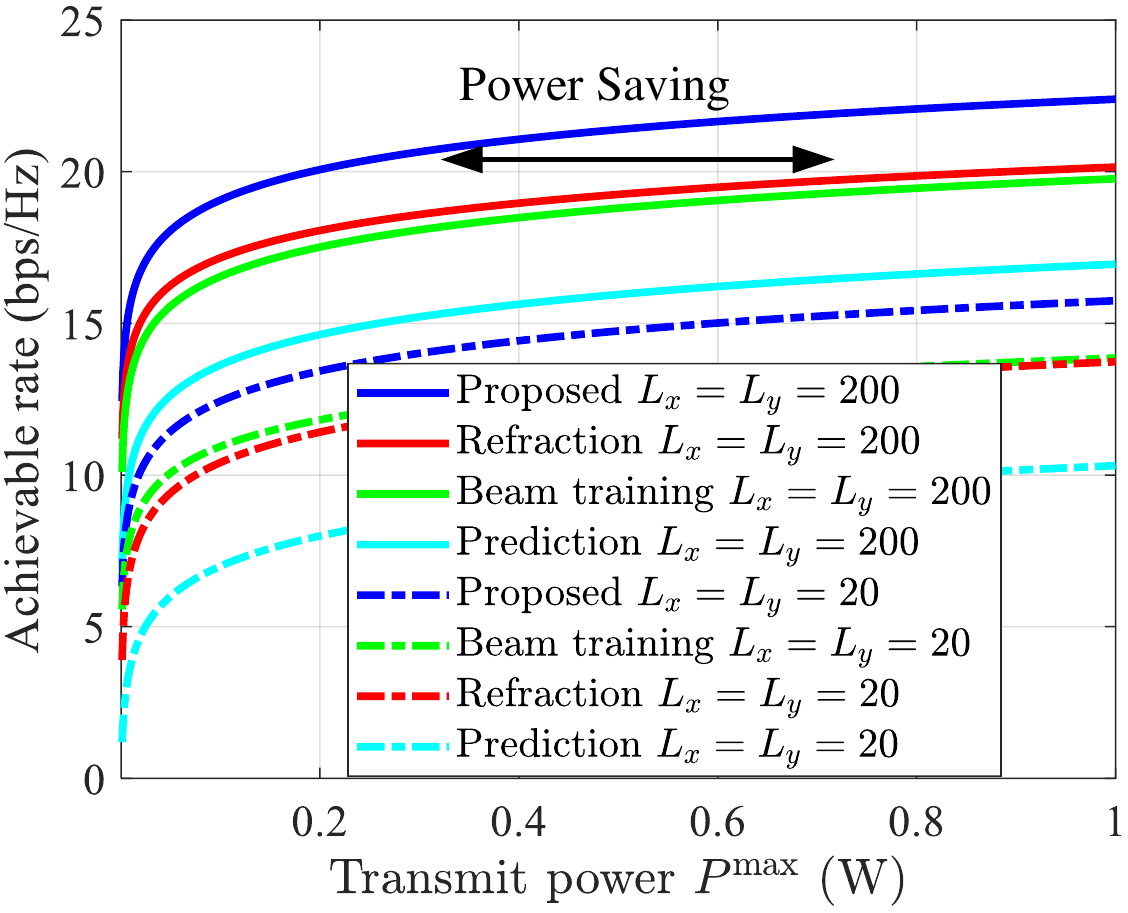}
	\vspace{-2mm}
	\caption{Average achievable rate of the whole period under different transmit power for single vehicle cases.}
	\label{figure9a}
\end{figure}

In Fig.~\ref{figure9a}, the average rate during the whole period is compared under different transmit power and different numbers of IOS elements for single-vehicle cases. Specifically, the rate of the proposed scheme is significantly improved compared to the benchmark schemes, mainly because the optimal power- and time-splitting ratios can effectively enlarge the sensing-assisted communication performance improvement. Moreover, to achieve the same rate, the proposed scheme can significantly reduce the transmit power requirements. For instance, to achieve a 20 bps/Hz rate, the transmit power of the proposed scheme can be reduced by at least four times compared with the benchmark schemes. The main reason is that proper power and time allocation for sensing can effectively improve the ability of beam alignment even in lower transmit power regions, thereby providing a larger beamforming gain during the communication-only stage. Therefore, deploying IOS on vehicles can effectively reduce the requirements of transmit power, and also improve both S\&C coverages of the RSU. 

To evaluate the impact of vehicle speed, the average rate during the whole period is compared under different vehicle speed. To ensure that the channel state can be assumed  approximately constant within a time slot, we set the length of the time slots according to the vehicle speeds, i.e., $\Delta T \propto  \frac{1}{v_{speed}}$. It can be seen from Fig.~\ref{figure9c} that the proposed scheme, refraction optimization scheme, and prediction scheme achieve worse communication performance with an increased velocity. The main reason is that within a shorter time slot, the remaining communication for pure communication will be reduced under the same sensing slot length. On the other hand, shorter time slot lengths lead to more limited sensing resources, which thus further reduces the average rate of the whole trajectory. 

\begin{figure}[t]
	\centering
	\includegraphics[width=7cm]{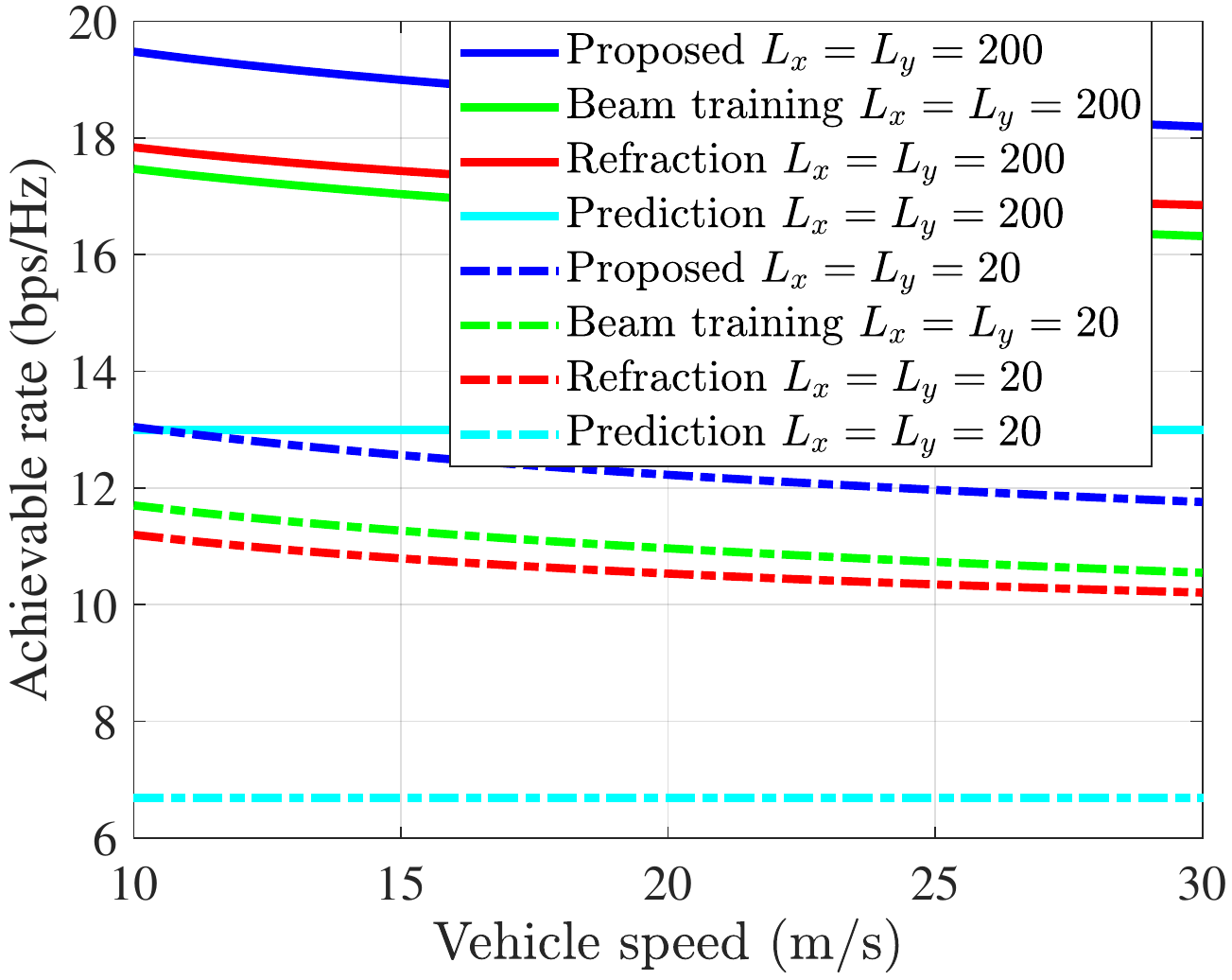}
	\vspace{-2mm}
	\caption{Communication performance comparison versus vehicle speed.}
	\label{figure9c}
\end{figure}

\begin{figure}[t]
	\centering
	\includegraphics[width=7cm]{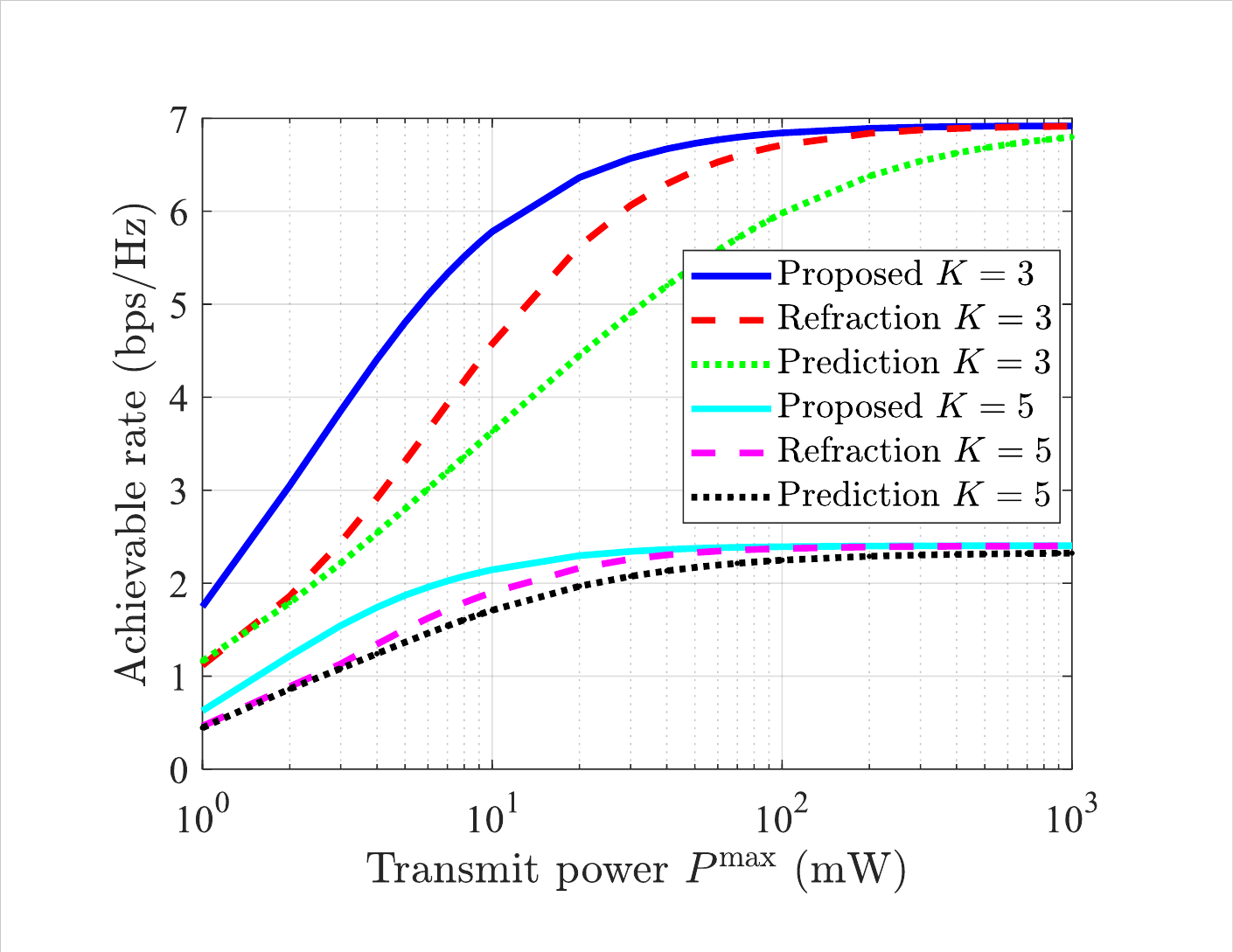}
	\vspace{-2mm}
	\caption{Average rate under different transmit power for multi-vehicle cases.}
	\label{figure9b}
\end{figure}
\begin{figure}[t]
	\centering
	\includegraphics[width=7cm]{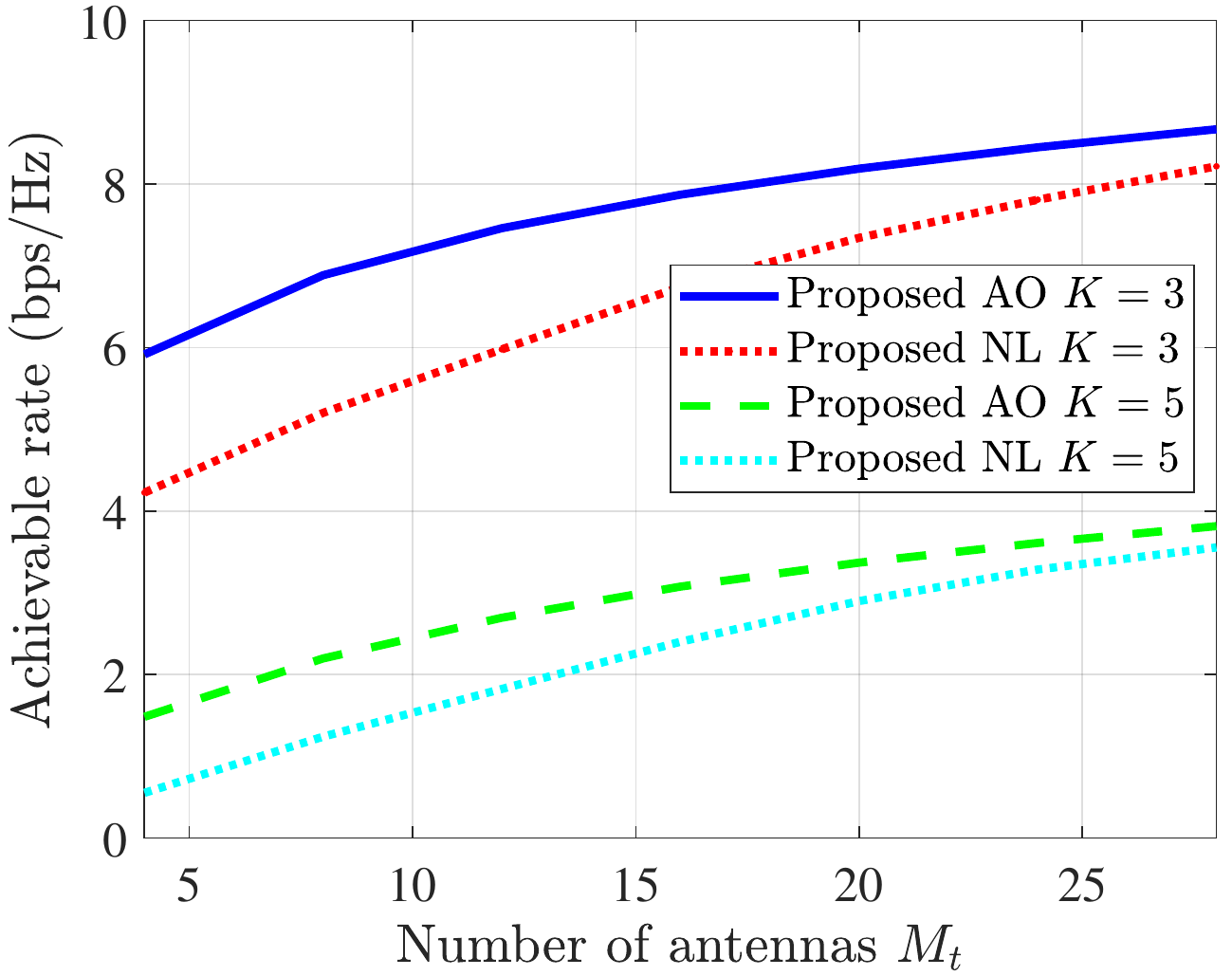}
	\vspace{-2mm}
	\caption{Comparison of the average rate under different the number of antennas.}
	\label{figure10}
\end{figure}

\subsection{Achievable Rate Comparison For Multi-Vehicle Cases}
In Fig.~\ref{figure9b}, the rate of the proposed scheme is compared with the benchmark schemes for multi-vehicle cases. The rate with $K$ = 5 is greatly reduced as compared to that with $K = 3$ due to more severe interference between the vehicles. It is worth noting that as the transmit power $P^{\max}$ increases to a saturation point, the rate of (P4) will not increase, which confirms the analysis in Proposition \ref{UpperBound}. For instance, with $K = 3$, the rate reaches the performance ceiling when the transmit power $P^{\max}$ is greater than 0.02 W, which is referred to as the transmit power saturation value. Moreover, it can be seen that the transmit power saturation value with $K = 3$ is larger than that with $K = 5$. The main reason is that the potential interference tends to be larger under more vehicles, and the corresponding noise can be ignored under smaller transmit power regions. 

Fig.~\ref{figure10} shows the performance comparison between the proposed AO algorithm in Section \ref{MultiVehicleClosedForm} (refer to "Proposed AO") and the proposed algorithm for NL cases in Section \ref{SectionNLCase} (refer to "Proposed NL"). Specifically, as the number of antennas increases, the performance gap between these two algorithms decreases. Intuitively, the beamwidth tends to be narrower and the potential interference between vehicles would be reduced as the number of antennas increases. On the other hand, the performance ceiling will be mainly affected by noise and the impact of interference would be relatively negligible when the number of antennas is sufficiently large.

\vspace{-2mm}
\section{Conclusions and Future Works}
\label{Conclusion}
In this paper, we investigated a new type of IOS-enabled sensing-assisted communication system and proposed a two-stage ISAC protocol to fulfil a more efficient balance and synergy between S\&C. The transmit beamforming, the IOS phase shift, and the duration of the joint S\&C stage were jointly optimized to maximize the minimum rate among the vehicles. A closed-form expression of the rate under uncertain locations was derived to significantly facilitate resource allocation of the considered systems. Furthermore, a sufficient and necessary condition for the existence of the joint S\&C stage is presented to further simplify the problem. Moreover, we propose a AO algorithm by exploiting the characteristics of performance limitation. The numerical results validated the efficiency of our design over the benchmark schemes and also confirmed the benefits of the sensing-assisted communication framework. The more general cases considering the effects caused by imperfectly compensated Doppler arbitrarily are worthwhile future works.

\vspace{-2mm}
\normalsize 
\section*{Appendix A: \textsc{Proof of Lemma \ref{OptimalPhaseShift}}}
It can be readily proved that the phase shift optimization problem for maximizing the SNR of the echo signals in (\ref{SensingReceivedPower}) is equivalent to the problem of the IOS beamforming gain maximization at each time slot since $\gamma^{{\mathrm{S}},x}_{k,n} = \frac{\eta \Delta T}{\Delta t \sigma_s^2} {\beta^2_{G,k,n}} | ({\bm{v}}^x_{k,n})^{H}	{\bm{b}}_{\mathrm{RSU}}(\varphi_{k,n}) |^2  |\bm{a}^{T}_{\mathrm{IOS}}(\phi_{k,n},-\varphi_{k,n})$  ${\bm{\Theta}}^{{\mathrm{S\&C,R}}}_{k,n} \bm{a}_{\mathrm{IOS}}(\phi_{k,n},-\varphi_{k,n} )|^2  |{{\bm{a}}_{\mathrm{RSU}}^T(\varphi_{k,n}){\bm{f}}_{k,n}^{\mathrm{S\&C}}}|^{2} $. 
Then, the IOS beamforming gain maximization problem is expressed as
\begin{equation}
	\mathop {\max }\limits_{{\bm{\Theta}}^{\xi, {\mathrm{R}}}_{n}} \  |{\bm{a}^{T}_{\mathrm{IOS}}\left(\phi_{n},-\varphi_{n}\right) {\bm{\Theta}}^{\xi, {\mathrm{R}}}_{n} \bm{a}_{\mathrm{IOS}}\left(\phi_{n},-\varphi_{n} \right)}|^2.
\end{equation}
We have ${\bm{a}^{T}_{\mathrm{IOS}}(\phi_{n},-\varphi_{n}) {\bm{\Theta}}^{\xi, {\mathrm{R}}}_{n} \bm{a}_{\mathrm{IOS}}(\phi_{n},-\varphi_{n} )} = {\bm{\theta}}^{\xi, {\mathrm{R}}}_{n} {\bm{a}}_{\mathrm{IOS}}\left(\phi_{n},-\varphi_{n}\right) \odot {\bm{a}}_{\mathrm{IOS}}\left(\phi_{n},-\varphi_{n}\right)$ according to $\mathbf{A}^T {\bm{X}} \mathbf{B}={\bm{x}}^T(\mathbf{A} \odot \mathbf{B})$ with ${\bm{X}} = {\rm{diag}}({\bm{x}})$, where ${\bm{\theta}}^{\xi, {\mathrm{R}}}_{n} = [\theta^{\xi, {\mathrm{R}}}_{1,n},\cdots,\theta^{\xi, {\mathrm{R}}}_{L,n}]$ and $\odot$ denotes the Hadamard product.
The received power at the RSU is maximized when the echo reflected from the IOS is directed towards the RSU, i.e., $\theta^{\mathrm{R}}_{(l_x-1)L_y+l_y}= - 2 \pi(l_x-1)  \cos(\varphi_{n|n-1})+ 2 \pi(l_y-1) \cos(\phi_{n|n-1})+\theta_{0}$, where {$\theta_{0}$} is the reference phase at the origin of the coordinates. Similarly, for the phase shift optimization problem of maximizing the SNR of the communication device, the optimal phase shift can be given by  $\theta^{\mathrm{T}}_{(l_x-1)L_y+l_y}=\pi(l_y-1) (\sin ( \psi^{u,z}_{n}) \cos( \psi^{u,x}_{n})- \cos(\varphi_{n|n-1}) )+\pi(l_y-1)(\cos(\phi_{n|n-1}) + \sin ( \psi^{u,z}_{n}) \sin( \psi^{u,x}_{n}))+\theta_{0}$. Thus, the proof is completed.

\vspace{-2mm}
\normalsize 
\section*{Appendix B: \textsc{Proof of Lemma \ref{EqualPowerSplitting}}}
\vspace{-1mm}
Without loss of generality, at the optimal solution of (P2), the total reflected power is assumed to be $\sum\nolimits_{l = 1}^L { {\beta _{l,n}^{\xi, {\mathrm{R}}}} } = X_R$. According to Cauchy-Buniakowsky-Schwartz inequality, the total received power at the RSU satisfies the following conditions:
\vspace{-2mm}
\begin{equation}\label{IOSemelemtsEquality}
	\begin{aligned}
		&{\left( {\sum\nolimits_{l = 1}^L {\sqrt {\beta _{l,n}^{\xi, {\mathrm{R}}}} } } \right)^2} \\
		=& \sum\nolimits_{l = 1}^L {\beta _{l,n}^{\xi, {\mathrm{R}}}}  + 2\sum\nolimits_{l = 1}^L {\sum\nolimits_{l' \ne l}^L {\sqrt {\beta _{l,n}^{\xi, {\mathrm{R}}}\beta _{l',n}^{\xi, {\mathrm{R}}}} } } \\
		\le &  \sum\nolimits_{l = 1}^L {\beta _{l,n}^{\xi, {\mathrm{R}}}}  + \sum\nolimits_{l = 1}^L \sum\nolimits_{l' \ne l}^L \left( {\beta _{l,n}^{\xi, {\mathrm{R}}} + \beta _{l',n}^{\xi, {\mathrm{R}}}} \right) \\
		=& L\sum\nolimits_{l = 1}^L {\beta _{l,n}^{\xi, {\mathrm{R}}}}   = L X_R,
		\vspace{-2mm}
	\end{aligned}
\end{equation}
where the equality in (\ref{IOSemelemtsEquality}) holds if and only if $\beta_l^{\xi, {\mathrm{R}}} = \beta_{l'}^{\xi, {\mathrm{R}}}$, $\forall l, l' \in {\cal{L}}$. Thus, with any given total reflected power, the SNR of the echo signals is maximized when $\beta_l^{\xi, {\mathrm{R}}} = \beta_{l'}^{\xi, {\mathrm{R}}}$. Similarly, for power ratio of the refracted signals through the IOS, $\beta_l^{\xi, {\mathrm{T}}} = \beta_{l'}^{\xi, {\mathrm{T}}}$, $\forall l, l' \in {\cal{L}}$.

\vspace{-2mm}
\normalsize 
\section*{Appendix C: \textsc{Proof of Proposition \ref{InftyBand}}}
Before proving Proposition \ref{InftyBand}, we first introduce a Lemma to facilitate the derivation of the closed-form SNR.
\begin{thm}\label{FejerKernel}
	Let $\sigma_{M}(x)=\frac{1}{2 a} \int_{-a}^{a} g(x+u) \tilde F_{M}(u)(u) d u$, where $\tilde F_{M}(u) = \frac{1}{M}\left(\frac{\sin \frac{N \pi u}{2a}}{\sin \frac{\pi u}{2a}}\right)^2$. If $g(x)$ is a real valued, continuous function with period $2a$, $\sigma_M(x)$ converges uniformly to $g(x)$ when $M \to \infty$, i.e.,
	\vspace{-2mm}
	\begin{equation}
		\frac{1}{2 a} \int_{-a}^{a} g(x+u) \tilde F_{M}(u) d u-g(x) \frac{1}{2 a} \int_{-a}^{a} \tilde F_{M}(u) d u = 0.
		\vspace{-2mm}
	\end{equation}
\end{thm}
\begin{proof}
	The detailed proof can be refereed to Theorem 3.6 in \cite{rust2013convergence}.
\end{proof}

According to Lemma \ref{FejerKernel}, we have $\lim _{M \rightarrow \infty} \frac{1}{2 a} \int_{-a}^a g(x+u) \tilde F_{M}(u) d u \rightarrow g(x)$, and thus, it follows that 
\vspace{-2mm}
\begin{equation}\label{EquationLemmaFejer}
	\lim _{M \rightarrow \infty} \frac{1}{2 a} \int_{-a}^a g(u) \tilde F_{M}(u) d u \rightarrow g(0).
	\vspace{-2mm}
\end{equation}

Furthermore, let $y = 2\cos \left( {{ \varphi_{n|n-1}}} \right) - 2\cos \left( \varphi_n  \right)$, where ${ \varphi_{n|n-1}} = { \varphi_{n}} + \omega_\varphi$ and $\omega_\varphi \in \mathcal{C} \mathcal{N}(0, {\sigma^2_{\omega_\varphi}} )$, then the probability density function (PDF) of $y$ is given in (\ref{ProbablityEquation}), as shown at the top of the next page.
\begin{figure*}
\begin{align}\label{ProbablityEquation}
	\vspace{-2mm}
	P_\varphi(y) &= \sum\limits_{i = -\infty}^{\infty} {G}{\left( {2(i + 1)\pi  -  {\arccos \left( {{{\frac{y}{2}}} + \cos \left( {{\varphi_n}} \right)} \right) - {\varphi_n}}} \right)^\prime } - {G}{\left( {2i\pi  + \arccos \left( {{\frac{y}{2}} + \cos \left( {{\varphi_n}} \right)} \right) - {\varphi_n}} \right)^\prime }  \nonumber \\
	&= \sum\limits_{i = -\infty}^{\infty} \frac{1}{{\sqrt {1 - {{\left( {{\frac{y}{2}} + \cos \left( {{\varphi_n}} \right)} \right)}^2}} }}\frac{1}{2{\sqrt {2\pi } {\sigma_{\omega_\varphi}}}}\left( {{e^{ - \frac{{{{\left( {2(i + 1)\pi  -  {\arccos \left( {{\frac{y}{2}} + \cos \left( {{\varphi_n}} \right)} \right)}  - {\varphi_n}} \right)}^2}}}{{2{\sigma^2_{\omega_\varphi}}}}}} + {e^{ - \frac{{{{\left( {2i\pi  + \arccos \left( {y + \cos \left( {{\varphi_n}} \right)} \right)} - \varphi_n \right)}^2}}}{{2{\sigma^2_{\omega_\varphi}}}}}}} \right),
\end{align}
\end{figure*}
In (\ref{ProbablityEquation}), ${G}(\cdot)$ is the cumulative distribution function (CDF) of Gaussian function. Define \[H(y) = \left\{ {\begin{array}{*{20}{c}}
		{{F_{{M_t}}}\left( {\frac{y}{2}} \right){F_{{M_r}}}\left( {\frac{y}{2}} \right){P_\varphi }(y),}&{y \in {\cal{Y}}}\\
		0&{{\rm{otherwise}}}
\end{array}} \right.,\]
where $ {\cal{Y}} = \left[ { - 1 - \cos {\varphi _n},1 - \cos {\varphi _n}} \right]$, and $\tilde{H}(y)$ is real valued, continue function with period $2$. Let $\tilde{H}(y) = \sum\limits_{i = -\infty}^{\infty} H(y - 2i)$, Then, it follows that
\begin{equation}\label{EquationProbability}
\begin{aligned}
	&\mathbb{E}_{\varphi}\left[ F_{L_x}( 2\Delta \cos \varphi_n)F_{M_t}(\Delta \cos \varphi_n) F_{M_r}(\Delta \cos \varphi_n)\right]	\\
	= &	\int\limits_{ - 2 - 2\cos {\varphi_n}}^{2 - 2\cos {\varphi_n}} {  \left( F_{L_x}\left( y\right)F_{M_t}\left(\frac{y}{2}\right) F_{M_r}\left(\frac{y}{2}\right)\right) P_\varphi(y)dy} \\
	\stackrel{(d)}\approx & \int\limits_{ - 1 - \cos {\varphi_n}}^{1 - \cos {\varphi_n}} { \tilde{H}(y)   \tilde F_{L_x}\left( y\right)dy} \\
	\stackrel{(e)}= & 2 \tilde{H}\left(0\right) = 2 M_t M_r P_\varphi(0),
\end{aligned}
\end{equation} 
where $F_{M}(x) = \frac{1}{M}\left(\frac{\sin \frac{M \pi x}{2 }}{\sin \frac{\pi x}{2 }}\right)^{2}$ and $P_\varphi(0) \!=\!  \sum\limits_{i = -\infty}^{\infty} \!\frac{1}{2{\sqrt {2\pi {\sigma^2_{\omega_\varphi}} \sin^2({\varphi_n})} }}\! \left({{e^{ - \frac{{{{\left( {2i\pi} \right)}^2}}}{{2{\sigma^2_{\omega_\varphi}}}}}} \!+\! {e^{ - \frac{{{{\left( {2(i + 1)\pi  } - 2\varphi_n \right)}^2}}}{{2{\sigma^2_{\omega_\varphi}}}}}}} \!\right)\!$. In (\ref{EquationProbability}), ($d$) holds since $H(y) \approx 0$ for $y \notin [-1 - \cos(\varphi_n), 1- \cos(\varphi_n)]$, and ($e$) holds based on (\ref{EquationLemmaFejer}), i.e., let $a = 1$ in Lemma \ref{FejerKernel}. Similarly, $\mathbb{E}_{ \phi}\left[ F_{L_y}\left( 2\Delta \cos \phi_n\right)\right] \!\approx\!\!  \sum\limits_{i = -\infty}^{\infty} \!\frac{1}{{\sqrt {2\pi {\sigma^2_{\omega_\phi}} \sin^2({\phi_n})} }}\left( {{e^{ - \frac{{{{\left( {2i\pi} \right)}^2}}}{{2{\sigma^2_{\omega_\phi}}}}}} \!+\! {e^{ - \frac{{{{\left( {2(i + 1)\pi  } - 2\phi_n \right)}^2}}}{{2{\sigma^2_{\omega_\phi}}}}}}} \!\right)\!$, and thus the expectation of the received echo's power can be given by 
\begin{equation}
	\gamma^{{\mathrm{S}},x}_n = \beta^{\mathrm{R}}\eta \frac{  \Delta T {P^{\max}}  \beta_{G,n}^2  M_t M_r L  h(\varphi_n, \sigma^2_{\omega_\varphi}) h(\phi_n, \sigma^2_{\omega_\phi}) }{{\sigma_s^2} } ,
\end{equation}
where 
$h(x, y) \!=\!  \frac{1}{{\sqrt {2\pi {y}\sin^2({x})} }} \sum\limits_{i = -\infty}^{\infty} \!\bigg(\!{e^{ - \frac{{{{2( {i\pi} )}^2}}}{{{y}}}}} + {e^{ - \frac{{{{2\left( {(i + 1)\pi  } - x \right)}^2}}}{{{y}}}}}\!\bigg)\!$. Similarly, it can be readily proved that $\gamma^{{\mathrm{S}},y}_n \approx \gamma^{{\mathrm{S}},x}_n$, and thus completes the proof.

\section*{Appendix D: \textsc{Proof of Proposition \ref{ClosedFormAchievable}}}
According to Lemmas \ref{OptimalPhaseShift} and \ref{EqualPowerSplitting}, the SNR at the communication device during the joint S\&C stage can be given by
\begin{align}\label{AchievableRateCommunication}
	\gamma^{\mathrm{S\&C}}_n \!= & (1-\!\beta^{\mathrm{R}})\frac{ P^{\max} {\beta_{G,n}} \beta_{h} L M_t }{{\sigma_{c}^{2}}}\mathbb{E}_{\phi}\left[F_{L_y}\left( \Delta \cos \phi_n\right)  \right]  \\
	& \times \mathbb{E}_{\varphi}\left[ F_{L_x}\left( \Delta \cos \varphi_n\right)F_{M_t}\left(\Delta \cos \varphi_n\right) \right]  \nonumber \\
	=& \left(1 - \beta^{\mathrm{R}} \right)\frac{ 4P^{\max}  {\beta_{G,n}} \beta_{h}   L M_t }{\sigma_{c}^{2}} h(\varphi_n, \sigma^2_{\omega_\varphi}) h(\phi_n, \sigma^2_{\omega_\phi}). \nonumber
\end{align}
During the communication-only stage, all the signal will be refracted to the direction of the device for communication improvement, i.e., $\beta^{\mathrm{C,R}} = 0$ and $\beta^{\mathrm{C,T}} = 1$.
The proof is similar to that of Proposition \ref{InftyBand}, for which the details are omitted for brevity.

\vspace{-1mm}
\section*{Appendix E: \textsc{Proof of Theorem \ref{SensingCondition}}}
\vspace{-1mm}
\label{EqualProblemP1}
It is not difficult to verify that with any given $\eta$, $\hat R_n$ defined in (\ref{AchievableRateApproximate}) is concave about $\beta$. Similarly, with any given $\beta^{\mathrm{R}}$, $\hat R_n$ is concave about $\eta$. Then, $\hat R$ can be given as a function w.r.t. $\eta$ and $\beta^{\mathrm{R}}$, i.e.,
\vspace{-1mm}
\begin{align}
	\hat R_n(\eta ,\beta^{\mathrm{R}} ) \! =& \eta {\log _2}\left( {1 +  C_1 {\left(1-\beta ^{\mathrm{R}}\right)}} \right) \\
	&+ ( {1 \!- \eta } ){\log _2}\!\left(\! {1 \!+\! D_1\sqrt {(\eta \beta _k^{\mathrm{R}} +\! D_2)(\eta \beta _k^{\mathrm{R}} +\! D_3)} } \!\right)\!	, \nonumber
	\vspace{-1mm}
\end{align}
where $D_1 = {{C_1}\sqrt {\frac{{\sigma _{{\omega _\varphi }}^2\sigma _{{\omega _\phi }}^2}}{{{A_{\varphi_n} }{A_{\phi_n} }}}} }$, $D_2 = \frac{{{A_{\varphi_n} }}}{{\sigma _{{\omega _\varphi }}^2}}$, and $D_3 = \frac{{{A_{\phi_n} }}}{{\sigma _{{\omega _\phi }}^2}}$. Then, the partial derivative of $\hat R_n(\eta ,\beta^{\mathrm{R}} )$ w.r.t. $\eta$ can be given by $\frac{\partial \hat R_n(\eta ,\beta^{\mathrm{R}} )}{\partial \eta} 
=  {\log _2}\left( 1 + { C_1 {\left(1-\beta ^{\mathrm{R}}\right)}} \right)+ \frac{{D_1\left( {1 - \eta } \right)\left( {2{{\left( {{\beta ^{\mathrm{R}}}} \right)}^2}\eta  + \left( {{D_2} + {D_3}} \right){\beta ^{\mathrm{R}}}} \right)}}{{2\sqrt {\left( {\eta {\beta ^{\mathrm{R}}} + {D_2}} \right)\left( {\eta {\beta ^{\mathrm{R}}} + {D_3}} \right)} \left( {1 + D_1\sqrt {\left( {\eta {\beta ^{\mathrm{R}}} + {D_2}} \right)\left( {\eta {\beta ^{\mathrm{R}}} + {D_3}} \right)} } \right)\ln 2}} - {\log _2}\left( {1 + D_1\sqrt {\eta {\beta ^{\mathrm{R}}} + {D_2}} \sqrt {\eta {\beta ^{\mathrm{R}}} + {D_3}} } \right) $.
With any given $\beta^{\mathrm{R}}$, if $ \left. {\frac{\partial \hat R_n(\eta ,\beta^{\mathrm{R}} )}{\partial \eta}} \right|_{\eta=0} \le 0$ always holds, we have $\hat R_n(\eta ,\beta^{\mathrm{R}} ) \le \hat R_n(0 ,\beta^{\mathrm{R}} )$, and $\eta^* = 0$ at the optimal solution of (P3). In the following, it will be proved $ \left. {\frac{\partial \hat R_n(\eta ,\beta^{\mathrm{R}} )}{\partial \eta}} \right|_{\eta=0} \le 0$ always holds if and only if $\frac{{\sigma _{{\omega _\phi }}^2}}{{{A_{\phi_n} }}} + \frac{{\sigma _{{\omega _\varphi }}^2}}{{{A_{\varphi_n} }}} \le 2$. When $\eta = 0$, the partial derivative of $\hat R_n(\eta ,\beta^{\mathrm{R}} )$ w.r.t. $\eta$ is given by
\vspace{-1mm}
\begin{equation} 
	\begin{aligned}
		&\left. {\frac{\partial \hat R_n(\eta ,\beta^{\mathrm{R}} )}{\partial \eta}} \right|_{\eta=0} \\
		= & {\log _2}\left( {1 + C_1 {\left(1-\beta ^{\mathrm{R}}\right)}} \right) - {\log _2}\left( {1 + D_1\sqrt {{D_2}{D_3}} } \right)  \\
		&+ \frac{{D_1\left( {{D_2} + {D_3}} \right){\beta ^{\mathrm{R}}}}}{{2\sqrt {{D_2}{D_3}} \left( {1 + D_1\sqrt {{D_2}{D_3}} } \right)\ln 2}} \buildrel \Delta \over = g(\beta ^{\mathrm{R}}).	
	\end{aligned}
\end{equation}
It can be verified that $g(\beta ^{\mathrm{R}})$ is concave about $\beta^{\mathrm{R}}$. If $g'(0) \le 0$, $g(\beta ^{\mathrm{R}}) \le g(0) = 0$, i.e., $ \left. {\frac{\partial \hat R(\eta ,\beta^{\mathrm{R}} )}{\partial \eta}} \right|_{\eta=0} \le 0$ always holds, where
$	g'(\beta ^{\mathrm{R}}) = \frac{1}{{\ln 2}}\left(\frac{{ - {C_1}}}{{1+ C_1 {\left(1-\beta ^{\mathrm{R}}\right)}}} + \frac{{D_1\left( {{D_2} + {D_3}} \right)}}{{2\sqrt {{D_2}{D_3}} \left( {1 + D_1\sqrt {{D_2}{D_3}} } \right)}} \right)$. Therefore, when $g'(0) = \frac{1}{{\ln 2}}\left(\frac{{ - {C_1}}}{1+C_1} + \frac{{D_1\left( {{D_2} + {D_3}} \right)}}{{2\sqrt {{D_2}{D_3}} \left( {1 + D_1\sqrt {{D_2}{D_3}} } \right)}}\right) \le 0$, we have $\frac{{\sigma _{{\omega _\phi }}^2}}{{{A_{\phi_n} }}} + \frac{{\sigma _{{\omega _\varphi }}^2}}{{{A_{\varphi_n} }}} \le 2$, and thus $\eta^* = 0$ at the optimal solution of (P3) since $ \left. {\frac{\partial \hat R(\eta ,\beta^{\mathrm{R}} )}{\partial \eta}} \right|_{\eta=0} \le 0$ always hold in this case. Otherwise, there always exist $\beta^{\mathrm{R}}$ making $g(\beta ^{\mathrm{R}}) > 0$, i.e., $\hat R(\eta ,\beta^{\mathrm{R}} ) > \hat R(0 ,\beta^{\mathrm{R}} )$ for $\eta > 0$. In this case, at the optimal solution of (P3), $\eta^* > 0$. Thus, the proof is completed.

\vspace{-1mm}
\section*{Appendix F: \textsc{Proof of Lemma \ref{InterferenceThm}}}
\vspace{-1mm}
Similar to the proof of Proposition \ref{InftyBand}, let $y =  \cos(\varphi_{k,n}) - \cos(\varphi_{j,n|n-1})$. When $M_t \to \infty$, it follows that
\begin{align}
&\mathbb{E}_{\varphi_{j,n|n-1}}\left[F_{M_t}\left( \cos(\varphi_{k,n}) -  \cos(\varphi_{j,n|n-1})\right)  \right] \\
=& \int\nolimits_{ - 1 - \cos {\varphi_{k,n}}}^{1 - \cos {\varphi_{k,n}}} {  \left( F_{L_x}\left( y\right)F_{M_t}\left({y}\right) F_{M_r}\left({y}\right)\right) P_\varphi(y)dy} \nonumber\\
=& \frac{1}{{\sqrt {2\pi {\sigma^2_{\omega_{\varphi}}} \sin^2({\varphi_{k,n}})} }} \sum\limits_{i = -\infty}^{\infty} \bigg( {e^{ - \frac{{{{\left( {2i\pi} + \varphi_{k,n} - \varphi_{j,n|n-1} \right)}^2}}}{{2{\sigma^2_{\omega_{\varphi}}}}}}} \nonumber \\
&+ {e^{ - \frac{{{{\left( {2(i + 1)\pi  } - \varphi_{k,n} - \varphi_{j,n|n-1}\right)}^2}}}{{2{\sigma^2_{\omega_{\varphi}}}}}}}\bigg),\nonumber
\vspace{-2mm}
\end{align}
and thus completes the proof.

\vspace{-2mm}
\normalsize 
\section*{Appendix G: \textsc{Proof of Proposition \ref{UpperBound}}}

Before proving Proposition \ref{UpperBound}, we first introduce a Lemma to facilitate its derivation. 

\begin{thm}\label{PowerCeiling}
	When the maximum transmit power $P^{\max}$ is larger than a saturation point, the optimal communication of (P4) will not increase when $P^{\max}$ increases.
\end{thm}
\begin{proof}
	First, it can be readily proved that the optimal achievable rate is a non-decreasing function about $P^{\max}$, since the optimal solution under $P^{\max}$ always satisfies all the constraints of (P4) under $P^{\max} + \Delta P$ with $\Delta P > 0$. When the transmit power is sufficiently large, (P4) is equivalent to (P4.2) since the noise is negligible. Without loss of generality, the saturation value of the transmit power is denoted by $\tilde P^{\max}$ when (P4) is equivalent to (P4.2). Under a given transmit power $P^{\max} > \tilde P^{\max}$, the optimal solution of (P4.2) is denoted by $\{{\bm{p}}^{S\&C*}, {\bm{p}}^{C*}, {\bm{\beta}}^{{\mathrm{R}}*}, \eta^*\}$. With another transmit power $\rho P^{\max}$ larger than the saturation point, where $\rho \le 1$, we can construct another solution $\{\rho{{\bm{p}}^{S\&C*}}, \rho{\bm{p}}^{C*}, {\bm{\beta}}^{{\mathrm{R}}*}, \eta^*\}$, the corresponding SIR in (\ref{ClosedFormAchievableInterferenceLImited}) and (\ref{ClosedFormAchievableCInterferenceLImited}) is equal to that of the solution $\{{\bm{p}}^{S\&C*}, {\bm{p}}^{C*}, {\bm{\beta}}^{{\mathrm{R}}*}, \eta^*\}$. Thus, when the maximum transmit power $P^{\max}$ is larger than a saturation point, the optimal communication of (P4) will not increase when $P^{\max}$ increases.
\end{proof}

Under the high transmit power region, $\eta R  ^{\mathrm{IL, S\&C}}_{k,n} +  \left( {1 - \eta } \right) R^{\mathrm{IL,C}}_{k,n} \le R^{\mathrm{IL,C}}_{k,n}$ due to $R^{\mathrm{IL, S\&C}}_k \le R^{\mathrm{IL, C}}_k$. Then, we further find the upper bound of $\mathop {\max }\limits_{k} \ R^{\mathrm{IL,C}}_{k,n}$. It can be readily proved that, at the optimal solution of $\mathop {\max }\limits_{k} \ R^{\mathrm{IL,C}}_{k,n}$, the optimal SIR of each communication device is equal, i.e., $\gamma^{\mathrm{C}}_k = \gamma^{\mathrm{C}}_j$. Then, it follows that
\vspace{-1mm}
\begin{equation}\label{EquationSloving}
	\frac{{p_k^{\mathrm{C}} M_t}}{{\sum\nolimits_{j \ne k}^K {p_j^{\mathrm{C}}{F_{k,j}}} }} = \frac{P^{\max} M_t}{{\sum\nolimits_{k = 1}^K {\sum\nolimits_{j \ne k}^K {p_j^{\mathrm{C}}{F_{k,j}}} } }}, \forall k \in {\cal{K}}, 
	\vspace{-1mm}
\end{equation}
and
\vspace{-1mm}
\begin{equation}
	{\sum\nolimits_{k = 1}^K {p_j^{\mathrm{C}}}}  = P^{\max}.
	\vspace{-1mm}
\end{equation}
Then, the optimal solution of $\mathop {\max }\limits_{k} \ R^{\mathrm{IL,C}}_{k,n}$ is given by
\vspace{-1mm}
\begin{equation}
	\bar R = {\log _2}\left( {1 + \frac{p^{\max} M_t}{{\sum\nolimits_{k = 1}^K {\sum\nolimits_{j \ne k}^K {p_j^{\mathrm{C}}{F_{M_t}\left(\cos(\hat \varphi_{k,n}), \cos(\hat \varphi_{j,n})\right)}} } }}} \right), 
	\vspace{-1mm}
\end{equation}
where $p_j^{\mathrm{C}}$ is obtained by solving the equations in (\ref{EquationSloving}). Then, we have 
\vspace{-1mm}
\begin{equation}
	\begin{aligned}
	&\frac{p^{\max} M_t}{{\sum\nolimits_{k = 1}^K {\sum\nolimits_{j \ne k}^K {p_j^{\mathrm{C}}{F_{M_t}\left(\cos(\hat \varphi_{k,n}), \cos(\hat \varphi_{j,n})\right)}} } }} \\
	\le& \frac{p^{\max} M_t}{{\mathop {\min }\limits_k \sum\nolimits_{j \ne k}^K {{F_{M_t}\left(\cos(\hat \varphi_{k,n}), \cos(\hat \varphi_{j,n})\right)}} }\sum\nolimits_{k = 1}^K p_k^{\mathrm{C}}} 	\\
	=& \frac{ M_t}{{\mathop {\min }\limits_k \sum\nolimits_{j \ne k}^K {{F_{M_t}\left(\cos(\hat \varphi_{k,n}), \cos(\hat \varphi_{j,n})\right)}} }},
	\vspace{-1mm}
	\end{aligned}
\end{equation}
and thus completes the proof.

\footnotesize  	
\bibliography{mybibfile}
\bibliographystyle{IEEEtran}

\end{document}